\documentclass[journal]{IEEEtran}
\newcommand{\RNum}[1]{\uppercase\expandafter{\romannumeral #1\relax}}
\usepackage{mathrsfs}
\usepackage{amssymb}
\usepackage[mathcal]{euscript}
\usepackage{diagbox}
\usepackage{cite}
\usepackage{graphicx}
\usepackage{subcaption}
\usepackage[T1]{fontenc}
\usepackage{graphicx}
\usepackage{CJKutf8}
\usepackage{xspace}
\usepackage{makecell}
\usepackage{psfrag}
\usepackage{stfloats}
\usepackage{amsmath}
\usepackage{array}
\usepackage{float}
\usepackage{multirow} 
\usepackage{amsthm}
\usepackage{color}
\usepackage{multicol}
\usepackage{stfloats}
\usepackage{enumerate}
\usepackage[utf8]{inputenc} 
\usepackage[T1]{fontenc}    
\usepackage{booktabs}       
\usepackage{amsfonts}       
\usepackage{nicefrac}       
\usepackage{microtype}      
\usepackage{color}
\usepackage{xcolor}
\usepackage{graphicx}
\usepackage{multirow}
\usepackage{soul}
\usepackage{mathrsfs}
\usepackage{bbm}
\usepackage{amssymb}

\usepackage[mathcal]{euscript}
\usepackage{psfrag,calc,bm}
\usepackage[hyphens]{url}
\usepackage[colorlinks,
            linkcolor=blue,
            anchorcolor=black,
            citecolor=blue
            ]{hyperref}
\usepackage{cite}

\usepackage{graphicx}

\usepackage{psfrag}

\usepackage{stfloats}

\usepackage{amsmath}

\usepackage{float}

\usepackage{algorithmic}
\usepackage[ruled,lined,linesnumbered]{algorithm2e}

\usepackage{color}

\usepackage{boxedminipage}

\usepackage{amsthm}

\usepackage{multirow}

\usepackage{setspace}

\usepackage{array}
\usepackage{wrapfig}
\usepackage{bbm}
\usepackage{mathtools}
\usepackage{dsfont}
\colorlet{soulgray}{lightgray!30}
\sethlcolor{soulgray}
\usepackage{booktabs}  
\newtheorem{remark}{Remark}
\newtheorem{example}{Example}
\newtheorem{theorem}{Theorem}

\allowdisplaybreaks[4]
\usepackage{fancyhdr}
\ifCLASSINFOpdf
\else
\fi
\begin{document}
\begin{CJK}{UTF8}{gbsn}
%
\title{Covert Prompt Transmission for Secure Large Language Model Services}


%
	\author{Ruichen Zhang, Yinqiu Liu, Shunpu Tang,  Jiacheng Wang,  Dusit Niyato,~\IEEEmembership{Fellow,~IEEE}, \\Geng Sun, Yonghui Li,~\IEEEmembership{Fellow,~IEEE}, and Sumei Sun,~\IEEEmembership{Fellow,~IEEE}



\thanks{R. Zhang, Y. Liu, S. Tang, J. Wang, and D. Niyato are with the College of Computing and Data Science, Nanyang Technological University, Singapore (e-mail: ruichen.zhang@ntu.edu.sg, yinqiu001@e.ntu.edu.sg, n2409411h@e.ntu.edu.sg,  jiacheng.wang@ntu.edu.sg, dniyato@ntu.edu.sg).}

\thanks{G. Sun is with the College of Computer Science and Technology, Jilin University, Changchun 130012, China (e-mail: sungeng@jlu.edu.cn).}

\thanks{Y. Li is with the School of Electrical and
Information Engineering, University of Sydney, Sydney, NSW 2006, Australia (e-mail: yonghui.li@sydney.edu.au).}

\thanks{S. Sun is with the Institute for Infocomm Research, Agency for Science, Technology and Research, Singapore (e-mail: sunsm@i2r.a-star.edu.sg).}

}
\maketitle
\begin{abstract}
This paper investigates covert prompt transmission for secure and efficient large language model (LLM) services over wireless networks. We formulate a latency minimization problem under fidelity and detectability constraints to ensure confidential and covert communication by jointly optimizing the transmit power and prompt compression ratio. To solve this problem, we first propose a prompt compression and encryption (PCAE) framework, performing surprisal-guided compression followed by lightweight permutation-based encryption. Specifically, PCAE employs a locally deployed small language model (SLM) to estimate token-level surprisal scores, selectively retaining semantically critical tokens while discarding redundant ones. This significantly reduces computational overhead and transmission duration. To further enhance covert wireless transmission, we then develop a group-based proximal policy optimization (GPPO) method that samples multiple candidate actions for each state, selecting the optimal one within each group and incorporating a Kullback-Leibler (KL) divergence penalty to improve policy stability and exploration. Simulation results show that PCAE achieves comparable LLM response fidelity to baseline methods while reducing preprocessing latency by over five orders of magnitude, enabling real-time edge deployment. We further validate PCAE effectiveness across diverse LLM backbones, including DeepSeek-32B, Qwen-32B, and their smaller variants. Moreover, GPPO reduces covert transmission latency by up to 38.6\% compared to existing reinforcement learning strategies, with further analysis showing that increased transmit power provides additional latency benefits.
\end{abstract}


\section{Introduction}

The rise of generative artificial intelligence (GAI) has led to the rapid development of large language models (LLMs), such as OpenAI’s ChatGPT and DeepSeek’s R1, which demonstrate strong capabilities in text generation, reasoning, and multimodal understanding~\cite{zhang2024generative}. Due to their high computational requirements, these models are typically hosted on centralized cloud servers, while user prompts are transmitted wirelessly from edge devices~\cite{10835069}. This cloud-edge architecture enables real-time services such as virtual assistants and customer support~\cite{10879580}, but also introduces significant security risks. In particular, transmitting prompts over open wireless channels exposes both their content and existence to potential eavesdropping or detection~\cite{10818760}. While prior research has primarily focused on improving LLM inference, the security risks associated with prompt transmission remain largely underexplored. Particularly, securing prompt transmission is crucial because many real-world prompts are extensive or semantically sensitive, including detailed financial transactions, comprehensive medical histories, and full-length document translations. These long prompts inherently increase transmission duration, directly escalating both the risks of content interception and the probability of adversarial detection~\cite{10818760}. Therefore, securing LLM-based wireless communication requires addressing two key challenges:

\textbf{\uppercase\expandafter{\romannumeral1}). Data Privacy and Query Confidentiality}:  
Wireless transmission of user queries to cloud-hosted LLMs inherently introduces privacy vulnerabilities, as adversaries may intercept, analyze, and infer sensitive information from transmitted signals~\cite{hanke2024open}.  
This risk is particularly critical in privacy-sensitive applications, where user prompts may contain personal, financial, or proprietary content, making them prime targets for identity theft or corporate espionage. Therefore, it is essential to apply encryption mechanisms that obfuscate semantic content at the token level, ensuring that even if signals are intercepted, meaningful information cannot be extracted. Such encryption not only secures the query content but also complements physical-layer stealth techniques, providing a foundational layer of protection for confidential LLM services.

\textbf{\uppercase\expandafter{\romannumeral2}). Transmission Security and Detection Avoidance}: Wireless transmission of LLM queries also poses the risk of detection, where adversarial wardens may monitor signal patterns to infer the existence of LLM-related traffic. Even if query content is fully encrypted, adversarial wardens can exploit transmission patterns, power fluctuations, and spectral characteristics to identify, classify, or block LLM communications. This challenge is further exacerbated by advanced detection techniques, such as channel state information (CSI)-based analysis, which can expose transmissions \cite{10633254}. To mitigate this risk, it is essential to ensure that LLM-related transmissions remain covert through adaptive power control, noise injection, and signal obfuscation strategies, making them indistinguishable from background noise.

To address the aforementioned security challenges, two important strategies have been introduced, i.e., prompt encryption and compression to safeguard query confidentiality and transmission signal covert communication to ensure undetectable wireless transmissions.

\textbf{\uppercase\expandafter{\romannumeral1}). LLM Prompt Encryption:} To address data privacy and query confidentiality, prompt encryption and compression have been introduced to transform user queries into an obfuscated and compressed representation, preventing adversaries from extracting meaningful content while preserving interpretability for LLM inference \cite{edemacu2024privacy}. Specifically, encryption techniques such as token-level substitution, structural transformation, and multimodal encoding have been applied to conceal sensitive information while ensuring efficient transmission. Additionally, lossy and lossless cryptographic methods are incorporated to minimize data size while retaining essential semantics. For example, numerical values and domain-specific terms in prompts can be transformed into structured placeholders (e.g., replacing ``The transaction of \$5000 was completed at Bank A'' with ``Transaction Amount at Institution X''), effectively anonymizing critical details without compromising the LLM's ability to generate responses \cite{buban2025encrypted}. 

\textbf{\uppercase\expandafter{\romannumeral2}). Transmission Signal Covert Communication:} To address transmission security and detection avoidance, transmission signal covert communication technique has been introduced to ensure that LLM-related wireless transmissions remain undetectable to adversarial monitoring. This approach prevents unauthorized entities from detecting, intercepting, or inferring the presence of LLM queries over wireless channels \cite{goeckel2015covert}. Specifically, covert communication techniques such as power control, artificial noise injection, and spread spectrum modulation have been employed to embed transmissions within background interference, thereby reducing the probability of detection. Unlike traditional encryption, which secures content but still reveals transmission activity, covert communication conceals the very existence of transmission, preventing adversarial wardens from even recognizing LLM-related communications \cite{an2024covert}. For example, leveraging artificial noise and stochastic power shaping enables transmitted signals to blend seamlessly into environmental noise, making them statistically indistinguishable from background interference.

Motivated by the aforementioned reasons, this work investigates covert prompt transmission for secure LLM services. Specifically, we develop a multi-layered security framework that integrates prompt encryption and compression with covert wireless transmission to protect both query confidentiality and transmission detectability. 
\textit{To the best of the authors' knowledge, this is the first work to introduce a covert prompt transmission method that ensures secure and undetectable LLM-based communications over wireless networks.} The key contributions of this paper are summarized as follows:
 
\begin{itemize}

\item \textbf{Problem Formulation and System Modeling:} 
{We formulate a covert prompt transmission problem that jointly optimizes the prompt compression ratio and transmit power to minimize end-to-end latency, while ensuring LLM response fidelity and maintaining the warden’s total detection error}. To this end, we derive a tractable expression for the detection probability constraint, effectively capturing the covertness requirement under wireless channel uncertainty.

\item \textbf{Prompt Compression and Encryption via PCAE:} To address the challenges of transmission overhead and query confidentiality, we propose a lightweight prompt compression and encryption (PCAE) framework that performs surprisal-based token selection using a local SLM, followed by permutation-based encryption to obfuscate semantic content before transmission.
    
\item \textbf{Covert Policy Optimization via GPPO:} To enable adaptive covert transmission, we design a group-based proximal policy optimization (GPPO) method that jointly optimizes the prompt compression ratio and transmit power. By evaluating multiple candidate actions per state and incorporating Kullback-Leibler (KL)-regularized updates, GPPO achieves stable and efficient policy learning under multi-factor constraints.
\end{itemize}
    
Simulation results show that the proposed PCAE framework and GPPO method achieve superior performance in both LLM response fidelity and covert communication efficiency. The PCAE framework effectively compresses and encrypts prompts, reducing the token length by over 40\% while preserving high response fidelity across a variety of LLM backbones. Compared with existing baselines, PCAE achieves comparable fidelity but reduces preprocessing latency by $10^5$, making it highly suitable for mobile edge deployment. For the covert transmission, the proposed GPPO method consistently outperforms state-of-the-art baselines, achieving significantly lower latency under strict fidelity constraints. Furthermore, GPPO reduces transmission latency by up to 38.6\% when the fidelity threshold $F_{\min}$ is set to 0.82.

\section{Related Work} \label{work}
In this section, we review related work on LLM prompt encryption and wireless covert communication and highlight key research gaps.

\subsection{LLM Prompt Encryption}
LLM prompt encryption aims to safeguard the confidentiality of user queries by transforming them into an obfuscated format that remains interpretable for LLM inference, ensuring privacy while preserving response quality \cite{das2025security}. Several recent studies have explored different encryption techniques to protect LLM queries. For example, Lin et al. \cite{lin2024emojicrypt} proposed EmojiCrypt, a prompt encryption technique that converted sensitive queries into emoji-based representations to prevent unauthorized access while maintaining interpretability. Similarly, Zhang et al. \cite{zhang2024secpe} introduced SecPE, a secure prompt ensembling framework that leveraged fully homomorphic encryption to protect LLM queries while minimizing computational overhead. In addition, Nazzal et al. \cite{nazzal2024promsec} developed PromSec, a prompt optimization framework utilizing generative adversarial graph neural networks to refine queries. Beyond text-based encryption, He et al. \cite{10847705} explored privacy-preserving sampling for latent diffusion models using homomorphic encryption, securing text embeddings. Meanwhile, Liang et al. \cite{liang2024my} investigated prompt leakage risks and proposed inference-time defenses based on prompt perplexity and token translation mechanisms, reducing extraction rates while maintaining model performance.

Despite recent progress, existing studies largely assume static cloud settings and neglect the unique challenges of wireless transmission, such as interception, efficiency constraints, and adversarial detection. Although some works consider prompt optimization in distributed or mobile scenarios, security issues remain insufficiently addressed. For instance, Liu et al. \cite{liu2025intelligent} proposed an intelligent mobile AIGC service that refines prompts via LLMs, but neglected transmission security risks. Likewise, Huang et al. \cite{huang2025joint} developed a joint optimization framework for edge-cloud systems, integrating an attack detector and Bayesian game modeling to enhance prompt security. However, while their approach improves security detection, it does not incorporate prompt encryption mechanisms to protect LLM queries from wireless interception threats.

\subsection{Wireless Covert Communication}
Unlike traditional encryption, covert communication ensures that transmissions blend into background noise, preventing detection and interception \cite{sobers2017covert}. To achieve covert transmission, researchers have explored various techniques. For example, Chen et al. \cite{chen2023covert} performed a comprehensive survey on wireless covert communications, providing insightful analysis about the challenges and solutions of developing covert communications. Wang et al. \cite{10680088} studied intelligent reflecting surface (IRS)-assisted covert communication in UAV air-ground networks, jointly optimizing beamforming and UAV positioning to improve covert transmission rates. Wu et al. \cite{10375580} proposed an IRS-assisted approach that adjusts transmit power, block length, and prior transmission probabilities to maximize the covert rate. Meanwhile, Hosseini et al. \cite{hosseini2023minimizing} investigated a framework minimizing the age of information (AoI) in time-varying channels while ensuring reliability. Yang et al. \cite{10086651} proposed an RS-IRS-NOMA-assisted covert system integrating rate-splitting, IRS, and NOMA to optimize secrecy outage probability. 

\begin{figure*}[tbp!]
  \centering
  \includegraphics[width=0.90\textwidth]{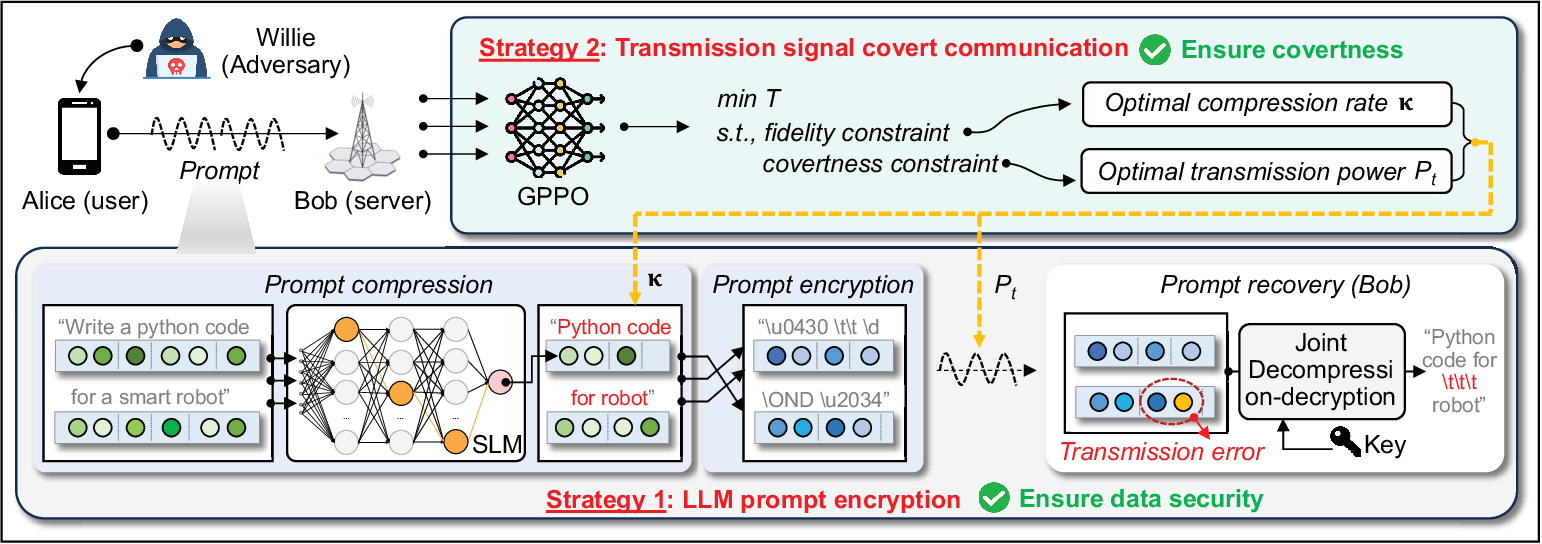} 
  \caption{System model of covert prompt transmission, where Alice compresses and encrypts the input prompt via an SLM before covertly transmitting it to Bob, while evading detection by Willie. We present prompt encryption to protect data security. Additionally, GPPO is trained to optimize compression rate and transmission power, thus realizing the covert transmission.}
  \vspace{-0.3cm}
  \label{System}
\end{figure*}

Despite these contributions, most existing works rely on optimization-based methods under idealized assumptions, which struggle in dynamic wireless environments~\cite{huang2021multi}. To address this, recent studies have adopted deep reinforcement learning (DRL) for more adaptive covert communication~\cite{10032267}. For example,  Li et al. \cite{li2023ddqn} proposed a deep double Q-learning network (DDQN)-based covert communication method for UAV swarms. Similarly, Bi et al. \cite{10122840} employed a time-averaged priority-based DDQN (TAP-DDQN) method in IRS-assisted UAV systems, optimizing UAV 3D trajectory and IRS phase shifts to enhance covert transmission performance. Xu et al. \cite{10681776} leveraged the joint proximal policy optimization method to optimize UAV positioning and power allocation in covert semantic communication, while Hu et al. \cite{10057472} adopted the twin delayed deep deterministic policy gradient method to maximize covert throughput. While existing DRL-based works improve transmission-level covert communication, they neglect LLM-driven services and content-level security, such as prompt encryption. Our work fills the gap by jointly addressing prompt encryption and covert transmission for secure LLM services over wireless networks.

\section{System Model}
As shown in Fig.~\ref{System}, we consider a covert prompt transmission system where a user (i.e., \textit{Alice}) sends queries to a cloud-based LLM (i.e., \textit{Bob}) over a wireless channel, while a passive adversary (i.e., \textit{Willie}) attempts to detect the transmission by distinguishing it from background noise.



\subsection{Prompt Encryption and Compression}
We consider an original user prompt query formulated in natural language, which is first tokenized into a sequence of tokens, i.e.,
\begin{equation}
    \mathbf{x} = [x_1, x_2, \dots, x_L],
\end{equation}
where $x_l$ denotes the $l$-th token, and $L$ is the length of the tokenized prompt sequence.

{In many practical LLM applications, such as document translation or legal contract summarization, the prompt can be significantly long, making direct transmission over bandwidth-limited wireless channels inefficient or infeasible. To address this, Alice performs prompt compression using a lightweight SLM~\cite{lin2024emojicrypt}.} This compression selectively retains the most informative tokens based on their semantic importance, yielding a compressed prompt sequence represented as
\begin{equation}
    \mathbf{x}_{\text{cmp}} = \mathcal{C}(\mathbf{x}; \kappa, \text{SLM}),
\end{equation}
where $\mathcal{C}(\cdot)$ is the SLM-based compression function and $\kappa \in (0,1]$ denotes the compression ratio. To further secure the prompt, an encryption method is applied, and the corresponding encrypted prompt sequence is expressed as 
\begin{equation}
    \mathbf{x}' = \mathcal{E}(\mathbf{x}_{\text{cmp}}; \mathbf{k}),
\end{equation}
where $\mathcal{E}(\cdot)$ represents the encryption function and $\mathbf{k}$ is a pre-shared secret key between Alice and Bob.  The final encrypted prompt is denoted as $\mathbf{x}' = [x'_1, x'_2, \dots, x'_{L'}]$, where $L' = \max(1,\lfloor\kappa L\rfloor)$ denotes the length of the encrypted sequence after compression.

To assess the effect of prompt processing, we denote the LLM responses to the original and encrypted-compressed prompts as $\mathbf{r} = \text{LLM}(\mathbf{x})$ and $\mathbf{r}' = \text{LLM}(\mathbf{x}')$, respectively. Since the transmitted prompt is both compressed and encrypted, preserving semantic fidelity is critical. To quantify this, we adopt BERTScore~\cite{Zhang2020BERTScore} to compute a response fidelity metric\footnote{The fidelity score $F_{\rm t}$ depends on both the compression ratio $\kappa$ and the encryption strategy. Lower values of $\kappa$ (i.e., fewer retained tokens) may lead to the loss of essential semantic content, thereby reducing $F_{\rm t}$.}, which is expressed as
\begin{equation}
    F_{\rm t} = \frac{1}{L_{\rm r}} \sum_{j=1}^{L_{\rm r}} \phi \left( \mathbf{e}(r_j), \mathbf{e}(r'_j) \right),
\end{equation}
where $\mathbf{e}(r_j)$ and $\mathbf{e}(r'_j)$ are the contextual embeddings of the $j$-th token in $\mathbf{r}$ and $\mathbf{r}'$, respectively, and $\phi(\cdot, \cdot)$ denotes cosine similarity. $L_{\rm r}$ is the number of tokens in the original LLM response.



\subsection{Wireless Transmission at Bob}  

The encrypted and compressed prompt sequence $\mathbf{x'}$ is first mapped to modulated symbols via a modulation function $\mathcal{M}(\cdot)$ and then transmitted over a wireless channel. The received signal at Bob is modeled as  
\begin{equation}
    y_{\rm b} = h_{\rm b} \sqrt{P_t} \psi + n_{\rm b},
\end{equation}  
where $\psi \sim \mathcal{CN}(0,1)$ denotes a normalized modulated symbol derived from $\mathbf{x'}$, $P_t$ is the transmit power, $h_{\rm b}$ is the complex channel coefficient between Alice and Bob, and $n_{\rm b} \sim \mathcal{CN}(0, \sigma_{\rm b}^2)$ represents the additive white Gaussian noise (AWGN) with variance $\sigma_{\rm b}^2$.

The wireless channel between Alice and each receiver (i.e., Bob or Willie) consists of large-scale and small-scale fading \cite{10680088}. The channel coefficient for a given receiver $i \in \{\text{b}, \text{w}\}$ is modeled as   $h_{i} = \sqrt{g_{i}} \tilde{h}_{i}$, where $g_{i}$ and $\tilde{h}_{i}$ represent the large-scale fading component and the small-scale fading component, respectively. For large-scale fading, it follows a distance-dependent path loss model, which is given by $g_{i} = \frac{g_0}{d_{i}^2}$, where $g_0$ is the reference channel gain at a unit distance and $d_{i}$ denotes the distance between Alice and receiver $i$. For small-scale fading, it follows a Rician fading model, which accounts for both a dominant line-of-sight (LoS) component and multipath scattering effects, i.e.,
\begin{equation}
    \tilde{h}_{i} = \sqrt{\frac{K}{K + 1}} + \sqrt{\frac{1}{K + 1}} \bar{h}_{i},
\end{equation}
where $K$ is the Rician factor indicating the power ratio between the LoS and scattered paths, and $\bar{h}_{i} \sim \mathcal{CN}(0,1)$ represents the Rayleigh-distributed multipath component.

The signal-to-noise ratio (SNR) at Bob is expressed as
\begin{equation}
    \text{SNR}_{\rm b} = \frac{|h_{\rm b}|^2 P_t}{\sigma_{\rm b}^2}. 
\end{equation}
Thus, the achievable transmission rate at Bob is given by
\begin{equation}
    R_{\rm b} = B \log_2(1 + \text{SNR}_{\rm b}),
\end{equation}
where $B$ denotes the system bandwidth. The total transmission time for the encrypted and compressed prompt sequence is given by
\begin{equation}
    L_{\text{T}} = \frac{L' S}{R_{\rm b}} + T_{\rm Proc},
\end{equation}
where $S$ is the average token size in bits, and $T_{\rm proc}$ denotes the processing time for the prompt compression and encryption via SLM.

\subsection{Binary Hypothesis Testing at Willie}
To detect Alice's transmission activity, Willie conducts a binary hypothesis test based on received signal energy. Under the null hypothesis $H_0$, Alice remains silent, and Willie observes only noise, while under the alternative hypothesis $H_1$, Alice is actively transmitting and Willie receives a superposition of signal and noise \cite{huang2025joint}. Over $T$ channel uses, the received signal at the Willie is expressed as
\begin{equation}
    y_{\rm w}[t] =
    \begin{cases}
        n_{\rm w}[t], & H_0, \\
        \sqrt{P_t} h_{\rm w} x'[t] + n_{\rm w}[t], & H_1,
    \end{cases}
\end{equation}
where \(y_{\rm w}[t]\) denotes the received signal sample at Willie, \(h_{\rm w}\) is the channel coefficient between Alice and Willie,  and \(n_{\rm w}[t] \sim \mathcal{CN}(0, \sigma_{\rm w}^2)\) is AWGN at Willie.

Using an energy detection strategy, the Willie computes the average received signal power over the $T$ samples, i.e.,
\begin{equation}
    T_{\rm w} = \frac{1}{T} \sum_{t=1}^{T} |y_{\rm w}[t]|^2.
\end{equation}
As $T \to \infty$, the strong law of large numbers implies that $T_{\rm w}$ converges to its expected value under each hypothesis:
\begin{equation}
    T_{\rm w} =
    \begin{cases} 
      \sigma_{\rm w}^2, & H_0, \\
      |h_{\rm w}|^2 P_t + \sigma_{\rm w}^2, & H_1.
    \end{cases}
\end{equation}
Willie then determines whether Alice is transmitting by comparing the measured power \(T_{\rm w}\) with a detection threshold \(\tau\), i.e., $T_{\rm w} \underset{H_0}{\overset{H_1}{\gtrless}} \tau$. If $T_{\rm w} > \tau$, the Willie declares $H_1$ (i.e., active transmission); otherwise, it decides $H_0$ (i.e., silence). The detection performance is characterized by the false alarm probability $P_{\rm FA}$ and missed detection probability $P_{\rm MD}$, which are respectively given by
\begin{equation}
P_{\rm FA} = \Pr(T_{\rm w} > \tau \mid H_0) = \Pr(\sigma_{\rm w}^2 > \tau),
\end{equation}
and
\begin{equation}
P_{\rm MD} = \Pr(T_{\rm w} \leq \tau \mid H_1) = \Pr(|h_{\rm w}|^2 P_t + \sigma_{\rm w}^2 \leq \tau).
\end{equation}
By controlling the transmission power \(P_t\) and hence the resulting received power at Willie, Alice can influence \(P_{\rm FA}\) and \(P_{\rm MD}\) to degrade Willie's detection probability and enhance the covertness of communication.

Under practical conditions, Willie may not have perfect knowledge of its noise power due to uncertainty in noise estimation \cite{hosseini2023minimizing}. We model the noise variance $\sigma_{\rm w}^2$ as a uniformly distributed random variable in the range as $\sigma_{\rm w}^2 \in \left[ \frac{\bar{\sigma}_{\rm w}^2}{\mu}, \bar{\sigma}_{\rm w}^2 \mu \right]$, where $\bar{\sigma}_{\rm w}^2$ is the nominal noise power and $\mu \geq 1$ is a known uncertainty parameter. The probability density function (PDF) of $\sigma_{\rm w}^2$ is given by
\begin{equation}
    f_{\sigma_{\rm w}^2}(x) =
    \begin{cases} 
      \frac{1}{2 \ln(\mu) x}, & \frac{\bar{\sigma}_{\rm w}^2}{\mu} \leq x \leq \bar{\sigma}_{\rm w}^2 \mu, \\
      0, & \text{otherwise}.
    \end{cases}
\end{equation}

Using this uncertainty model, the false alarm probability at Willie is given by
\begin{equation}
\label{Eq_detection_V1}
    P_{\rm FA} =
    \begin{cases} 
      1, & \tau < \frac{\bar{\sigma}_{\rm w}^2}{\mu}, \\
      \frac{1}{2 \ln(\mu)} \ln \left( \frac{\bar{\sigma}_{\rm w}^2 \mu}{\tau} \right), & \frac{\bar{\sigma}_{\rm w}^2}{\mu} \leq \tau \leq \bar{\sigma}_{\rm w}^2 \mu, \\
      0, & \tau > \bar{\sigma}_{\rm w}^2 \mu.
    \end{cases}
\end{equation}
Similarly, the missed detection probability is given by (\ref{Eq_detection}).
\begin{figure*}
\begin{equation}
\label{Eq_detection}
    P_{\rm MD} =
    \begin{cases} 
      0, & \tau < |h_{\rm w}|^2 P_t + \frac{\bar{\sigma}_{\rm w}^2}{\mu}, \\
      \frac{1}{2 \ln(\mu)} \ln \left( \frac{\mu(\tau - |h_{\rm w}|^2 P_t)}{\bar{\sigma}_{\rm w}^2} \right), & |h_{\rm w}|^2 P_t + \frac{\bar{\sigma}_{\rm w}^2}{\mu} \leq \tau \leq |h_{\rm w}|^2 P_t + \bar{\sigma}_{\rm w}^2 \mu, \\
      1, & \tau > |h_{\rm w}|^2 P_t + \bar{\sigma}_{\rm w}^2 \mu.
    \end{cases}
\end{equation}
\hrule
\end{figure*}

\begin{theorem}
Considering the log-uniform noise uncertainty model in our work, the optimal detection threshold $\tau^*$ that minimizes Willie's detection performance is given by
\begin{equation}
    \tau^* = |h_{\rm w}|^2 P_t + \frac{\bar{\sigma}_{\rm w}^2}{\mu}.
\end{equation}
At this threshold, the minimum total detection error probability is given by
\begin{equation}
    \xi^* = \frac{1}{2 \ln(\mu)} \ln \left( \frac{\mu \bar{\sigma}_{\rm w}^2}{|h_{\rm w}|^2 P_t + \bar{\sigma}_{\rm w}^2 / \mu} \right).
\end{equation}
\end{theorem}

\begin{proof}
Let $\xi = P_{\rm FA} + P_{\rm MD}$ denote the total detection error probability. Substituting the expressions of $P_{\rm FA}$ and $P_{\rm MD}$, we rewrite $\xi$ as in (\ref{Eq_Xi}).
\begin{figure*}
\begin{equation}
\label{Eq_Xi}
    \xi =
    \begin{cases} 
      1, & \tau < \frac{\bar{\sigma}_{\rm w}^2}{\mu}, \\
      \frac{1}{2 \ln(\mu)} \ln \left( \frac{\bar{\sigma}_{\rm w}^2 \mu}{\tau} \right), & \frac{\bar{\sigma}_{\rm w}^2}{\mu} \leq \tau < |h_{\rm w}|^2 P_t + \frac{\bar{\sigma}_{\rm w}^2}{\mu}, \\
      \frac{1}{2 \ln(\mu)} \ln \left( \frac{\bar{\sigma}_{\rm w}^2 \mu}{\tau} \right) + \frac{1}{2 \ln(\mu)} \ln \left( \frac{\mu(\tau - |h_{\rm w}|^2 P_t)}{\bar{\sigma}_{\rm w}^2} \right), & |h_{\rm w}|^2 P_t + \frac{\bar{\sigma}_{\rm w}^2}{\mu} \leq \tau \leq \bar{\sigma}_{\rm w}^2 \mu, \\
      \frac{1}{2 \ln(\mu)} \ln \left( \frac{\mu(\tau - |h_{\rm w}|^2 P_t)}{\bar{\sigma}_{\rm w}^2} \right), & \bar{\sigma}_{\rm w}^2 \mu < \tau \leq |h_{\rm w}|^2 P_t + \bar{\sigma}_{\rm w}^2 \mu, \\
      1, & \tau > |h_{\rm w}|^2 P_t + \bar{\sigma}_{\rm w}^2 \mu.
    \end{cases}
\end{equation}
\hrule
\end{figure*}
Observing the functional behavior of $\xi$, we note that it monotonically decreases in the range, i.e.,
\begin{equation}
    \frac{\bar{\sigma}_{\rm w}^2}{\mu} \leq \tau < |h_{\rm w}|^2 P_t + \frac{\bar{\sigma}_{\rm w}^2}{\mu}.
\end{equation}
Thus, the optimal detection threshold at Willie that minimizes his detection performance is given by $\tau^* = |h_{\rm w}|^2 P_t + \frac{\bar{\sigma}_{\rm w}^2}{\mu}$. At this threshold, the minimum total error probability is given by
\begin{equation}
    \xi^* = \frac{1}{2 \ln(\mu)} \ln \left( \frac{\mu \bar{\sigma}_{\rm w}^2}{|h_{\rm w}|^2 P_t + \bar{\sigma}_{\rm w}^2 / \mu} \right).
\end{equation}
To ensure covert communication, Alice must control the transmit power such that Willie's minimum detection error probability satisfies as $\xi^* \geq 1 - \epsilon$, which guarantees that Alice's transmission remains covert while maintaining reliable communication with Bob.
\end{proof}

\subsection{Problem Formulation}
We aim to minimize the transmission latency while ensuring reliable prompt reception at Bob and maintaining covertness against Willie\footnote{
While this work focuses on covert prompt transmission from Alice to the LLM server (Bob), we do not consider the reverse direction, i.e., Bob's response to Alice, since it is typically less sensitive to detectability. The downlink response is often transmitted after a noticeable delay and can be inherently embedded within regular background traffic, thereby posing a lower risk of adversarial detection. Similar mechanisms have also been adopted in existing covert communication works, e.g.,~\cite{chen2023covert,goeckel2015covert}.
}. To achieve this, we jointly optimize the compression ratio $\kappa$ and the transmit power $P_t$ under constraints on LLM response fidelity, transmission reliability, and covertness. The optimization problem is formulated as follows:
\begin{subequations}\label{P1}
    \begin{align}
    \min_{\{\kappa, P_t\}} \quad & L_{\text{T}}, \label{P1_a} \\
    \text{s.t.} \quad  
    & F_{\rm t} \geq F_{\min}, \label{P1_b} \\
    & 0 < \kappa \leq 1, \label{P1_c} \\
    & 0 \leq P_t \leq P_{\max}, \label{P1_d} \\
    & R_{\rm b} \geq R_{\min}, \label{P1_e} \\
    & \xi^* \geq 1 - \epsilon. \label{P1_f}
    \end{align}
\end{subequations}
The constraint in~\eqref{P1_b} guarantees that the LLM response fidelity remains above a required threshold $F_{\min}$, while the constraints in~\eqref{P1_c} and~\eqref{P1_d} ensure the compression ratio and transmit power within feasible ranges. The constraint in~\eqref{P1_e} ensures a sufficient transmission rate for reliable prompt decoding at Bob. Lastly, the constraint in~\eqref{P1_f} enforces the covertness condition by requiring the Willie’s total detection error to exceed $1 - \epsilon$.

The Problem~\eqref{P1} is non-convex due to the coupling between $\kappa$ and $P_t$ in both the latency objective and the covert constraint~\eqref{P1_f}, which involves nonlinear functions derived from statistical detection theory. To address this, we propose a unified approach that integrates prompt compression and encryption with covert transmission optimization. Specifically, an SLM-based compression module reduces token overhead while retaining essential semantics, and a GPPO method is employed to learn transmission strategies that minimize latency while satisfying fidelity and covertness requirements.

\section{Proposed PCAE Framework}
This section presents the proposed PCAE framework with SLM-based compression and permutation-based encryption for securing user queries during wireless transmission to cloud-based LLMs. 
\begin{figure}[t]
\centering
\includegraphics[width=0.45\textwidth, height=0.47\textwidth]{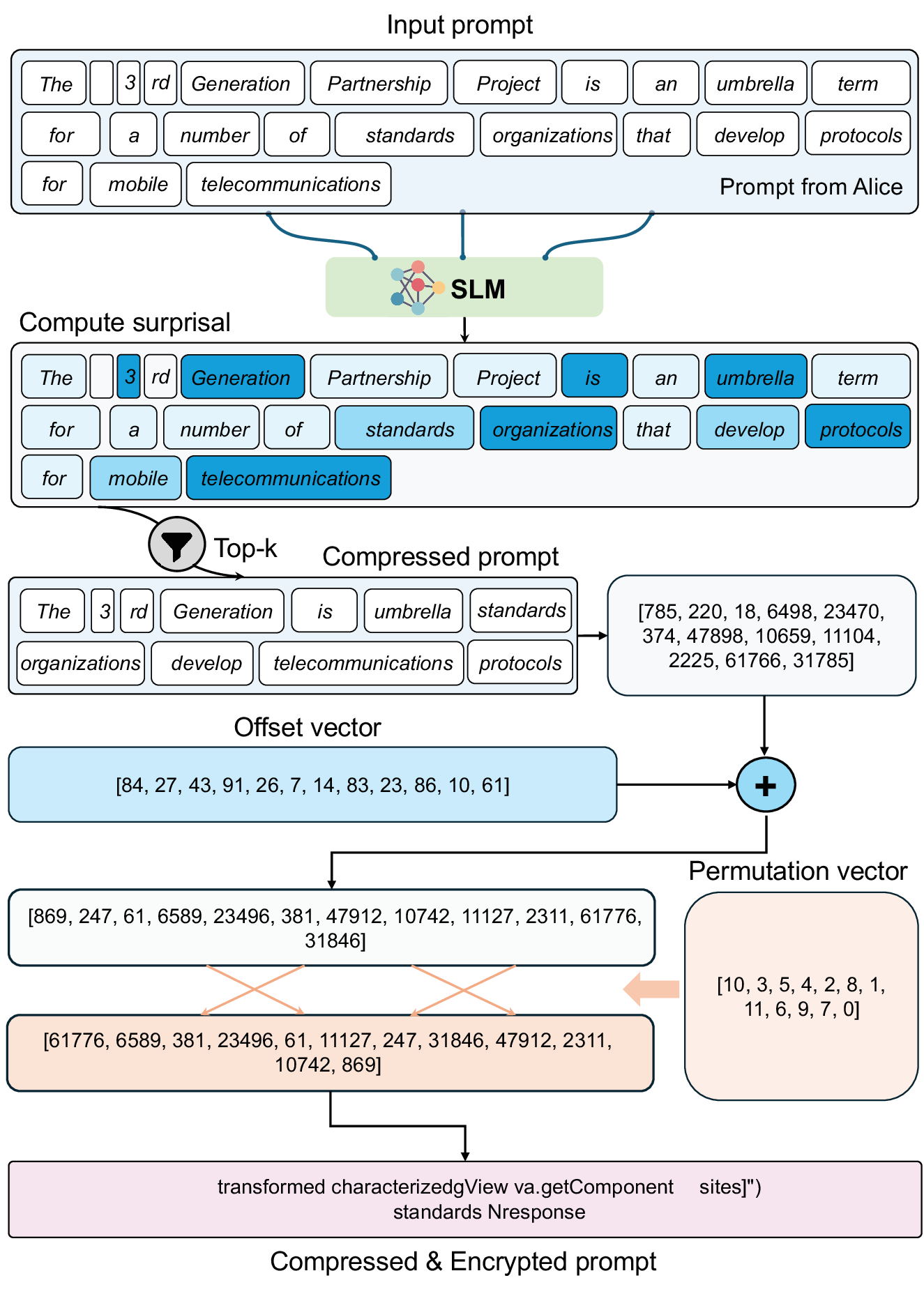}
\caption{Illustration of the PCAE framework, which combines surprisal-based prompt compression and lightweight token-level encryption to enable secure and efficient prompt transmission.}
\label{fig:gppo}
\end{figure}

\subsection{Prompt Compression}

{To reduce transmission overhead and enable covert transmission, we propose a lightweight prompt compression method based on a locally deployed SLM. The goal is to retain semantically critical tokens while discarding redundant content, ensuring that the compressed prompt still enables accurate response generation. By reducing prompt length, the compressed sequence shortens transmission time, thereby decreasing the adversarial Willie's window of detection and enhancing communication covertness.} The detailed design is as follows:  

\textbf{Surprisal-Based Token Scoring:}  
Inspired by Shannon's information theory\cite{shannon1948mathematical}, we initially assess token informativeness using entropy. For a token $x_l$ in the original prompt $\mathbf{x}$, the entropy is defined as
\begin{equation}
    H(x_l) = - \log \rho_{\text{prompt}}(x_l),
\end{equation}
where $\rho_{\text{prompt}}(x_l)$ denotes its empirical probability. However, this measure lacks contextual awareness\cite{oh2023does}. To better capture semantic importance in context, we leverage the autoregressive prediction capability of the SLM and define surprisal as
\begin{equation}
    s(x_l) = - \log p_{\text{SLM}}(x_l \mid x_1, x_2, \dots, x_{l-1}),
    \label{eq:neg_log}
\end{equation}
where $p_{\text{SLM}}(\cdot)$ is the conditional probability predicted by the SLM. A lower surprisal indicates that the token is highly predictable and semantically redundant, while higher surprisal reflects semantic significance and unpredictability.

\textbf{Computing Surprisal from Logits:}  
To compute surprisal, the prompt $\mathbf{x}$ is fed into the SLM to generate a sequence of logits, i.e.,
\begin{equation}
    \mathbf{z} = [\mathbf{z}_1, \mathbf{z}_2, \dots, \mathbf{z}_L] = \text{SLM}(\mathbf{x}),
\end{equation}
where $\mathbf{z}_l \in \mathbb{R}^{|\mathcal{V}|}$ is the unnormalized logit vector for position $l$. Applying softmax over the vocabulary yields
\begin{equation}
    p_{\text{SLM}}(x_l \mid x_1, \dots, x_{l-1}) =
    \frac{\exp(\mathbf{z}_l[x_l])}{\sum_{v \in \mathcal{V}} \exp(\mathbf{z}_l[v])}.
\end{equation}
Surprisal scores are then obtained as
\begin{equation}
    \mathbf{s} = [s(x_2), s(x_3), \dots, s(x_L)],
\end{equation}
excluding $x_1$ due to its undefined history.

\textbf{Token Selection:}  
Given a compression ratio $\kappa \in (0,1]$, the compressed prompt length is
\begin{equation}
    L' = \max(1, \lfloor \kappa L \rfloor).
\end{equation}
To preserve prompt semantics, we adopt a head-tail preservation strategy. Specifically, we retain the first $L_{\text{head}} = \min(N_{\rm h}, L)$ and the last $L_{\text{tail}} = \min(N_{\rm t}, L - L_{\text{head}})$ tokens. The remaining $L_{\text{keep}} = L' - L_{\text{head}} - L_{\text{tail}}$ tokens are selected from the middle segment based on the highest surprisal scores. The final compressed sequence is expressed as
\begin{equation}
    \mathbf{x}_{\text{cmp}} = \left[ x_1, \dots, x_{L_{\text{head}}}, \mathbf{x}_{\text{top}}, x_{L - L_{\text{tail}} + 1}, \dots, x_L \right],
\end{equation}
where $\mathbf{x}_{\text{top}}$ consists of the selected informative tokens, reordered to match their original positions. The compression procedure is summarized in Algorithm~\ref{alg:SLM}.

\begin{remark}
The proposed prompt compression method leverages the parallel nature of Transformer-based SLMs, enabling all surprisal scores to be computed in a single forward pass without iterative decoding \cite{zhang2024adaptive}. This significantly reduces inference latency and computational cost, {thus further shortening the overall prompt processing and transmission duration. Consequently, this prompt compression method enhances covert performance by minimizing the adversarial Willie's opportunity for detection, particularly beneficial for longer prompts.}
\end{remark}


\begin{algorithm}[t!]
    \caption{Prompt Compression with SLM.}
    \label{alg:SLM}
    \KwIn{Original prompt $\mathbf{x} = [x_1, x_2, \ldots, x_L]$, compression ratio $\kappa$, head/tail reserve limits $N_ {\rm h}$, $N_{\rm t}$, and \text{SLM}}
    \KwOut{Compressed prompt $\mathbf{x'}$}
    
    Tokenize $\mathbf{x} \rightarrow \mathbf{id} = [id_1, id_2, \ldots, id_L]$\;
    
    Compute target token count $L' \leftarrow \max(1, \lfloor \kappa L \rfloor)$\;
    
    Obtain logits $\mathbf{z} \leftarrow \texttt{SLM}(\mathbf{id})$\;
    
    Compute surprisal scores $s(x_l) = -\log \rho_{\text{SLM}}(x_l \mid x_1,\ldots,x_{l-1})$\;
    
    Reserve head and tail indices: $L_{\text{head}} = \min(N_{\rm h}, L), \quad L_{\text{tail}} = \min(N_{\rm t}, L - L_{\text{head}})$\;
    
    $\mathcal{I}_{\text{res}} \leftarrow \{1, \ldots, L_{\text{head}}\} \cup \{L - L_{\text{tail}} + 1, \ldots, L\}$\;
    
    $L_{\text{kep}} \leftarrow L' - |\mathcal{I}_{\text{res}}|$\;
    
    Select top-$L_{\text{kep}}$ tokens with highest surprisal:
    $\mathcal{I}_{\text{top}} \leftarrow \text{TopK}(\{s(x_l)\}, L_{\text{kep}}), \; l \notin \mathcal{I}_{\text{res}}$\;
    
    $\mathcal{I}_{\text{final}} \leftarrow \mathcal{I}_{\text{res}} \cup \mathcal{I}_{\text{top}}$\;
    
    Sort $\mathcal{I}_{\text{final}}$ in order and extract tokens $\mathbf{id'}$\;
    
    Decode $\mathbf{id'} \rightarrow \mathbf{x}_{\rm cmp}$\;
    
    \Return $\mathbf{x}_{\rm cmp}$\;
    
    \end{algorithm}

\subsection{Prompt Encryption}
After prompt compression, we propose a token-level encryption method based on offset and permutation operations, which effectively obfuscates the token sequence with minimal computational overhead. The detailed design is as follows:

\textbf{Offset-Based Token Obfuscation:}
Given the compressed prompt $\mathbf{x}_{\text{cmp}}$ and an SLM tokenizer with vocabulary size $V$, we first generate a random offset vector, i.e.,
\begin{equation}
\mathbf{o} = [o_1, o_2, \dots, o_{L}], 
\end{equation}
where each element is uniformly sampled from the integer set, i.e.,
\begin{equation}
    o_i \sim \mathcal{U}(r_{\min}, r_{\max}), \quad \forall i \in \{1, 2, \dots, |\mathcal{V}|\},
\end{equation}
with $r_{\min}$ and $r_{\max}$ denoting the predefined minimum and maximum offset values, respectively. Each token in the compressed sequence $\mathbf{x}_{\text{cmp}}$ is then obfuscated using the corresponding offset, which is expressed as
\begin{equation}
    x'_{\text{off}, l} = (x_{\text{cmp}, l} + o_{x_{\text{cmp}, l}}) \mod |\mathcal{V}|,
\end{equation}
where $x'_{\text{off}, l}$ denotes the $l$-th token of the offset token sequence $\mathbf{x}'_{\text{off}}$.

\textbf{Prompt Encryption with Permutation:}
To further enhance encryption strength, we apply a permutation operation that reorders the token positions in the offset-encrypted prompt~\cite{7295616}. Specifically, we define a permutation map as
\begin{equation}
    \mathcal{P} = [\pi(1), \pi(2), \dots, \pi(L')],
\end{equation}
where $\pi(l) \in \{1, 2, \dots, L'\}$ denotes the new position of the $l$-th token after permutation, and $\mathcal{P}$ is a random permutation of $\{1, \dots, L'\}$. 

The final encrypted prompt $\mathbf{x}' = [x'_1, x'_2, \dots, x'_{L'}]$ is generated by applying $\mathcal{P}$ to the offset-encrypted sequence $\mathbf{x}'_{\text{off}}$, i.e.,
\begin{equation}
    x'_l = x'_{\text{off}, \pi(l)}, \quad \forall l \in \{1, 2, \dots, L'\}.
\end{equation}
The resulting encrypted sequence $\mathbf{x}'$ is then ready for covert wireless transmission. The detailed encryption process is summarized in Algorithm~\ref{alg:encryption}.

\begin{remark}
The proposed encryption method is symmetric and relies on a shared secret key between Alice and Bob, which includes the offset vector $\mathbf{o}$ and permutation map $\mathcal{P}$. These are securely exchanged prior to communication using standard key exchange protocols. Notably, since the encryption process operates on token IDs and preserves their total count, it introduces no additional transmission overhead. The obfuscated sequence remains invertible, allowing Bob to accurately recover the original prompt via inverse operations.
\end{remark}

\begin{algorithm}[t!]
    \caption{Permutation-based Prompt Encryption.}
    \label{alg:encryption}
    \KwIn{Compressed prompt $\mathbf{x}_{\text{cmp}}$, offset range $[r_{\min}, r_{\max}]$, and vocabulary size $V$}
    \KwOut{Encrypted prompt $\mathbf{x'}$}
    
    Generate random offset vector $\mathbf{o} = [o_1, o_2, \ldots, o_{L'}]$\;
    
    Compute offset sequence: $x'_{\text{off}, l} = (x_{\text{cmp}, l} + o_{x_{\text{cmp}, l}}) \mod V$\;
    
    Generate permutation map $\mathcal{P}$\;
    
    Apply permutation: $x'_l = x'_{\text{off}, \mathcal{P}(l)}$\;
    \Return $\mathbf{x'}$\;
\end{algorithm}

\begin{example}
     For the sentence ``Private and Covert Communications'', the tokenized sequence is  [16787, 323, 3539, 1621, 25466], under the tokenizer of Qwen2.5. We set the offset vector as [4, 2, 9, 1, 4], and obtain the offset sequence as [16791, 325, 3548, 1622, 25470].  After applying the permutation map $\mathcal{P} =$ [2, 0,  4, 1,  3], the encrypted token sequence is [3548, 16791, 25470, 325, 1622], which means ``src\_attribute strictlyse dec" after decoding without the shared secret key.
\end{example}

\section{Proposed GPPO Approach} \label{DRL}
After ensuring prompt security via compression and encryption, we next consider covert wireless transmission. To this end, we propose the GPPO method to minimize transmission latency under covert communication constraints.

\subsection{MDP Design}
We model the covert prompt transmission problem as a Markov decision process (MDP), where the key components are defined as follows:

\textbf{State Space:}  
The state space captures key features of the wireless environment, enabling the agent to make informed decisions. Specifically, the state includes the transmission time $L_{\rm T}$, Bob's achievable transmission rate $R_{\rm b}$,  LLM response fidelity $F_{\rm t}$, and the detection error probability at the adversarial Willie $\xi^{*}$. Thus, the state space is formally represented as
\begin{equation}
    \mathcal{S} = \left\{L_{\rm T}, R_{\rm b}, F_{\rm t}, \xi^{*}\right\}.
\end{equation}
Consequently, the dimension of the state space is 4.

\textbf{Action Space:}  
The action space consists of two primary components, i.e., the prompt compression ratio $\kappa$ and the transmit power level $P_t$. For the prompt compression ratio $\kappa$, it is discretized into a finite set of $M$ predefined levels, i.e.,
\begin{equation}
    \mathcal{K} = \{\kappa_1, \kappa_2, \dots, \kappa_M\},
\end{equation}
where each $\kappa_m$ represents a distinct granularity of compression determined by the PCAE framework described in Section~III.

For the transmit power $P_t$, it is continuously adjusted and directly affects transmission quality and security. To ensure $P_t$ remains within the feasible range $0 \leq P_t \leq P_{\max}$, we employ a neural network with a hyperbolic tangent activation function. Thus, the transmit power is computed as
\begin{equation}
 P_t = \frac{P_{\max}}{2}\left(\underbrace{\frac{\exp(x^{\text{pow}}) - \exp(-x^{\text{pow}})}{\exp(x^{\text{pow}}) + \exp(-x^{\text{pow}})}}_{\text{hyperbolic tangent function}} + 1\right),
\end{equation}
where  $x^{\text{pow}}$ is the intermediate output of the neural network's decision layer. Thus, the action space is formally represented as
\begin{equation}
    \mathcal{A} = \left\{\{\kappa_{m}\}, P_{t}\right\}.
\end{equation}
Consequently, the dimension of the action space is $M+1$.

\begin{figure}[t]
\centering
\includegraphics[width=0.49\textwidth]{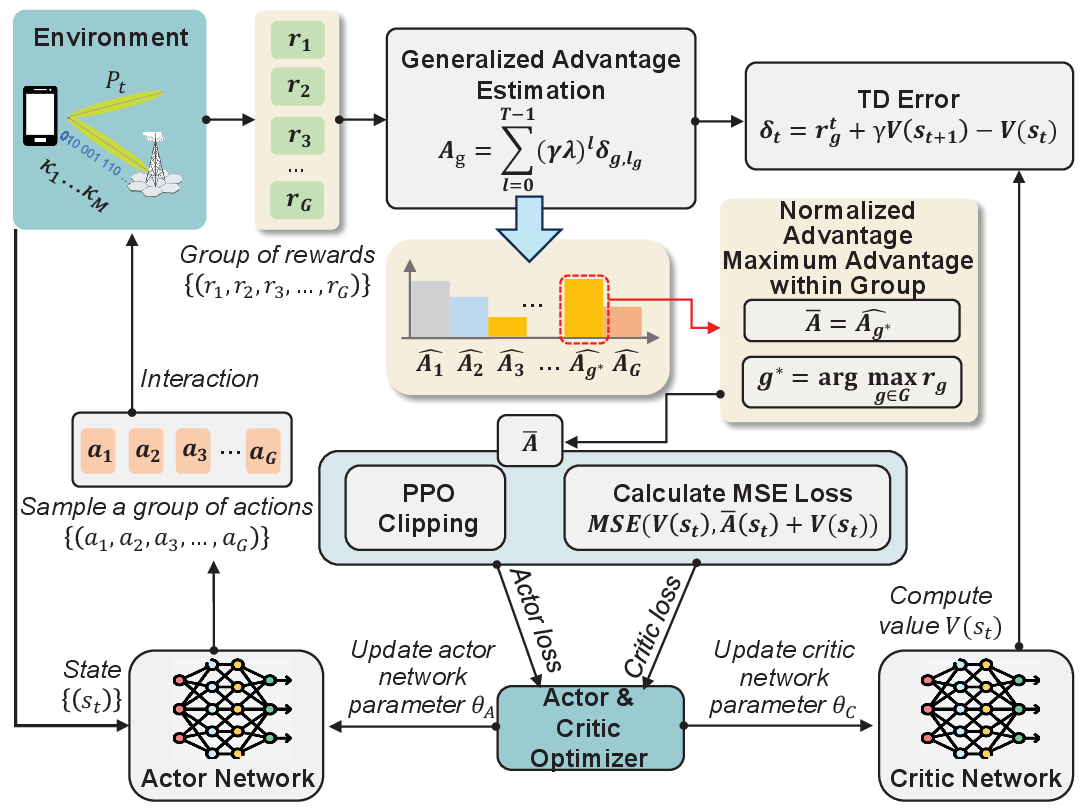}
\caption{Illustration of the proposed GPPO method, which integrates group-wise action sampling and critic-guided advantage estimation for stable and efficient covert transmission policy learning.}
\label{fig:gppo}
\end{figure}

\textbf{Reward Function:}  
The reward function is designed to guide the agent toward minimizing transmission latency while strictly satisfying covert communication constraints. To achieve this, we define the immediate reward as
\begin{equation}
r =
\begin{cases}
\dfrac{1}{L_{\rm T}}, & \text{if } \Lambda_{\rm t} = 0, \\[6pt]
0, & \text{otherwise},
\end{cases}
\end{equation}
where $p_{\rm t}$ denotes the aggregated penalty triggered when any of the key constraints, i.e., LLM response fidelity, transmission rate, or detection error probability, is violated. The penalty term is given by
\begin{equation}
\Lambda_{\rm t} = \frac{
\beta_{\rm r} (R_{\rm min} - R_{\rm b})^{+} +
\beta_{\rm e} (1 - \epsilon - \xi^{*})^{+} +
\beta_{\rm f} (F_{\rm min} - F_{\rm t})^{+}
}{
\beta_{\rm r} + \beta_{\rm e} + \beta_{\rm f}
},
\end{equation}
where $(x)^+ = \max(0, x)$ denotes the hinge function that activates the penalty only when the corresponding constraint is violated. The coefficients $\beta_{\rm r}$, $\beta_{\rm e}$, and $\beta_{\rm f}$ control the relative importance of transmission reliability, detection error margin, and LLM response fidelity, respectively.

\subsection{GPPO Framework}

PPO is a widely used actor-critic reinforcement learning algorithm that ensures stable policy updates by optimizing a clipped surrogate objective~\cite{10032267}. However, its direct application to covert LLM prompt transmission is limited by two factors: (i) PPO samples a single action per state, which may fail to capture complex reward dependencies arising from joint compression, encryption, and transmission; (ii) this single-sample strategy can lead to suboptimal convergence in high-dimensional action spaces. To address these challenges, as shown in Fig.~\ref{fig:gppo}, we propose GPPO, which extends PPO by sampling a group of $G$ candidate actions at each state and selecting the one with the highest reward for policy updates~\cite{shao2024deepseekmath}. This design promotes broader exploration, stabilizes learning in sparse-reward environments, and improves gradient estimation by leveraging the best-performing local decisions.

Let $\boldsymbol\theta_{\rm A}$ and $\boldsymbol\lambda_{\rm C}$ denote the parameters of the actor and critic networks, respectively. The stochastic policy $\pi_{\boldsymbol\theta_{\rm A}}(a|s)$ governs action selection, and $V_{\boldsymbol\lambda_{\rm C}}(s)$ estimates the value of each state.

\textbf{Advantage Estimation via GAE:}
Given a state $s$, GPPO samples a group of $G$ candidate actions $\{a_1, a_2, \ldots, a_G\}$, observing their corresponding immediate rewards and next states. For each action $a_g$, we first compute its temporal difference (TD) error, i.e.,
\begin{equation}
    \delta_{g,l_{g}} = r_g + \gamma V(s'_g) - V(s),
\end{equation}
where $r_g$ is the immediate reward after taking action $a_g$, and $s'_g$ is the resulting next state. The advantage value $A_g$ for each action $a_g$ is then calculated via generalized advantage estimation (GAE) \cite{10945973} as
\begin{equation}
    A_g = \sum_{l=0}^{T-1} (\gamma\lambda)^l \delta_{g,l_{g}},
\end{equation}
where $\gamma$ is the discount factor, $\lambda$ is the GAE decay parameter, and $\delta_{g,l_{g}}$ denotes the TD error at the $l$-th step following action $a_g$. Subsequently, we normalize these advantage estimates across all actions sampled at the current state to reduce variance and stabilize training, which is expressed as
\begin{equation}
   \hat{A}_{g} = \frac{A_g - \text{Mean}(\{A_g\})}{\text{Std}(\{A_g\})}.
\end{equation}
The group-level advantage $\bar{A}$ is then selected based on the action that yields the highest reward. Specifically, we have $\bar{A} = \hat{A}_{g^*}$,  where $g^* = \arg \max_{g\in \{1,\ldots,G\}} r_g$. This update mechanism ensures that policy learning is driven by the best local action, improving efficiency and convergence.

\textbf{Policy Optimization with KL Regularization:} To achieve stable policy updates, we utilize the clipped surrogate objective function derived from PPO, which restricts the magnitude of policy changes and prevents drastic updates, i.e.,
\begin{equation}\label{eq:gppo_clip}
\small
\mathcal{L}_{\rm GPPO}(\boldsymbol\theta_{\rm A}) = \mathbb{E}\left[ 
\min\left(\rho(\boldsymbol\theta_{\rm A}) \bar{A}, 
\text{clip}(\rho(\boldsymbol\theta_{\rm A}), 1 - \epsilon_{\rm{clip}}, 1 + \epsilon_{\rm{clip}}) \bar{A}
\right)
\right],
\end{equation}
where $\rho(\boldsymbol\theta_{\rm A}) = \frac{\pi_{\boldsymbol\theta_{\rm A}}(a|s)}{\pi_{\boldsymbol\theta_{\rm A}^{\rm old}}(a|s)}$ represents the probability ratio between the new and old policies, and $\epsilon$ is the clipping parameter that controls the extent of policy update.

To further enhance stability and promote exploration during training, we incorporate an additional KL divergence penalty term between the updated policy and the previous policy. Thus, the final GPPO objective function is expressed as
\begin{equation}\label{eq:gppo_kl}
\mathcal{L}_{\rm GPPO}^{\rm final}(\boldsymbol\theta_{\rm A}) = \mathcal{L}_{\rm GPPO}(\boldsymbol\theta_{\rm A}) - \beta \cdot D_{\rm KL}\left(\pi_{\boldsymbol\theta_{\rm A}}\|\pi_{\boldsymbol\theta_{\rm A}^{\rm old}}\right),
\end{equation}
where $\beta$ is a tunable hyperparameter controlling the strength of KL regularization. Specifically, the KL divergence is computed as
\begin{equation}
D_{\rm KL}\left(\pi_{\boldsymbol\theta_{\rm A}}\|\pi_{\boldsymbol\theta_{\rm A}^{\rm old}}\right) = \log\frac{\pi_{\boldsymbol\theta_{\rm A}}(a|s)}{\pi_{\boldsymbol\theta_{\rm A}^{\rm old}}(a|s)}.
\end{equation}

\textbf{Critic Network Update:}
The critic network is trained to provide accurate state-value estimations, serving as a baseline to reduce variance during policy optimization. Specifically, the parameters of the critic network $\boldsymbol\lambda_{\rm C}$ are updated by minimizing the mean squared error (MSE) loss between the estimated state-value function and the target value, which is defined as
\begin{equation}\label{eq:critic_mse}
\mathcal{L}_{\rm MSE}(\boldsymbol\lambda_{\rm C}) = \left[ V_{\boldsymbol\lambda_{\rm C}}(s) - V_{\text{target}}(s) \right]^2,
\end{equation}
where $V_{\boldsymbol\lambda_{\rm C}}(s)$ is the critic network's predicted value for state $s$, and $V_{\text{target}}(s)$ represents the target state value computed based on the observed returns. The critic network is optimized via gradient descent to ensure stable training and reliable advantage estimations. The detailed GPPO process in summarized Algorithm~\ref{alg:gppo_simple}.

\begin{algorithm}[t!]
\caption{Group-wise Proximal Policy Optimization (GPPO).}
\label{alg:gppo_simple}
\KwIn{Actor parameters $\boldsymbol\theta_{\rm A}$, critic parameters $\boldsymbol\lambda_{\rm C}$, group size $G$}
\KwOut{Updated $\boldsymbol\theta_{\rm A}$ and $\boldsymbol\lambda_{\rm C}$}

\For{each iteration}{
    \For{each exploration step}{
        Observe state $s$\;
        Sample $G$ candidate actions $\{a_1, \ldots, a_G\} \sim \pi_{\boldsymbol\theta_{\rm A}}(\cdot|s)$\;
        Execute each $a_g$, collect rewards $\{r_1, \ldots, r_G\}$ and next state $s'$\;
        Store $(s, a_{g^{*}}, r_{g^{*}}, s')$ in buffer\;
    }

    Compute advantage estimates $\hat{A}_g$ using GAE for each group\;

    \For{each group in buffer}{
        Identify best action: $g^* = \arg\max_{g} r_g$\;
        Set $\bar{A} = \hat{A}_{g^*}$, $a = a_{g^*}$\;

        Compute clipped surrogate loss and KL penalty as in Eq.~\eqref{eq:gppo_kl}\;

        Update actor parameters $\boldsymbol\theta_{\rm A}$ via gradient ascent\;

        Compute target value $V_{\text{target}}(s)$ and update $\boldsymbol\lambda_{\rm C}$ via MSE loss\;
    }
}
\Return Updated $\boldsymbol\theta_{\rm A}$ and $\boldsymbol\lambda_{\rm C}$
\end{algorithm}

\begin{remark}
Our proposed GPPO method differs fundamentally from the group relative policy optimization (GRPO) method~\cite{shao2024deepseekmath}. GRPO is optimized for large-scale autoregressive LLM environments and omits the critic network, relying solely on reward signals derived from generated text. In contrast, GPPO incorporates a critic network to estimate state values, which is particularly well-suited to our wireless communication setting, where actions directly impact both transmission reliability and covertness. Given the relatively low-dimensional state and action spaces, the critic network introduces negligible computational overhead while providing enhanced training stability. 
\end{remark}

\begin{table}[!t]
\renewcommand{\arraystretch}{1.0}
\setlength{\tabcolsep}{4pt}
\caption{Simulation and Training Parameters}
\label{tab:sim_settings}
\centering
\begin{small}
\begin{tabular}{|l|c|}
\hline
\textbf{Parameter} & \textbf{Value} \\ \hline \hline
Token bit length $S$ & 0.2 \\ \hline
Encryption time $T_{\rm comp}$ & 1 \\ \hline
Fidelity threshold $F_{\min}$ & 0.86 \\ \hline
Detection error threshold $\epsilon$ & 0.05 \\ \hline
Compression levels $M$ & 10 \\ \hline
GAE decay $\lambda$ & 0.97 \\ \hline
Activation function & ReLU \\ \hline
coefficients $\{\beta_{\rm r}, \beta_{\rm e}, \beta_{\rm f} \}$ &  $\{1,1,1\}$\\ \hline
Discount factor $\gamma$ & 0 \\ \hline
Training iterations & $5 \times 10^5$ \\ \hline
\end{tabular}
\end{small}
\end{table}

\subsection{Computational Complexity Analysis}
For GPPO, in each training step, the computational complexity is dominated by the forward and backward passes through the actor and critic neural networks. Assuming $P$ hidden layers and $n_p$ neurons in the $p$-th layer, the computational complexity is $\mathcal{O}(\sum_{p=1}^P n_{p-1} \cdot n_p)$ \cite{10032267}. Since GPPO evaluates $G$ actions per state via group sampling, the overall complexity scales linearly with $G$, yielding a total complexity of $\mathcal{O}(G \cdot \sum_{p=1}^P n_{p-1} \cdot n_p)$.

\section{Simulation} \label{sim}
\begin{figure*}[tbp!]
  \centering
  \includegraphics[width=1\textwidth]{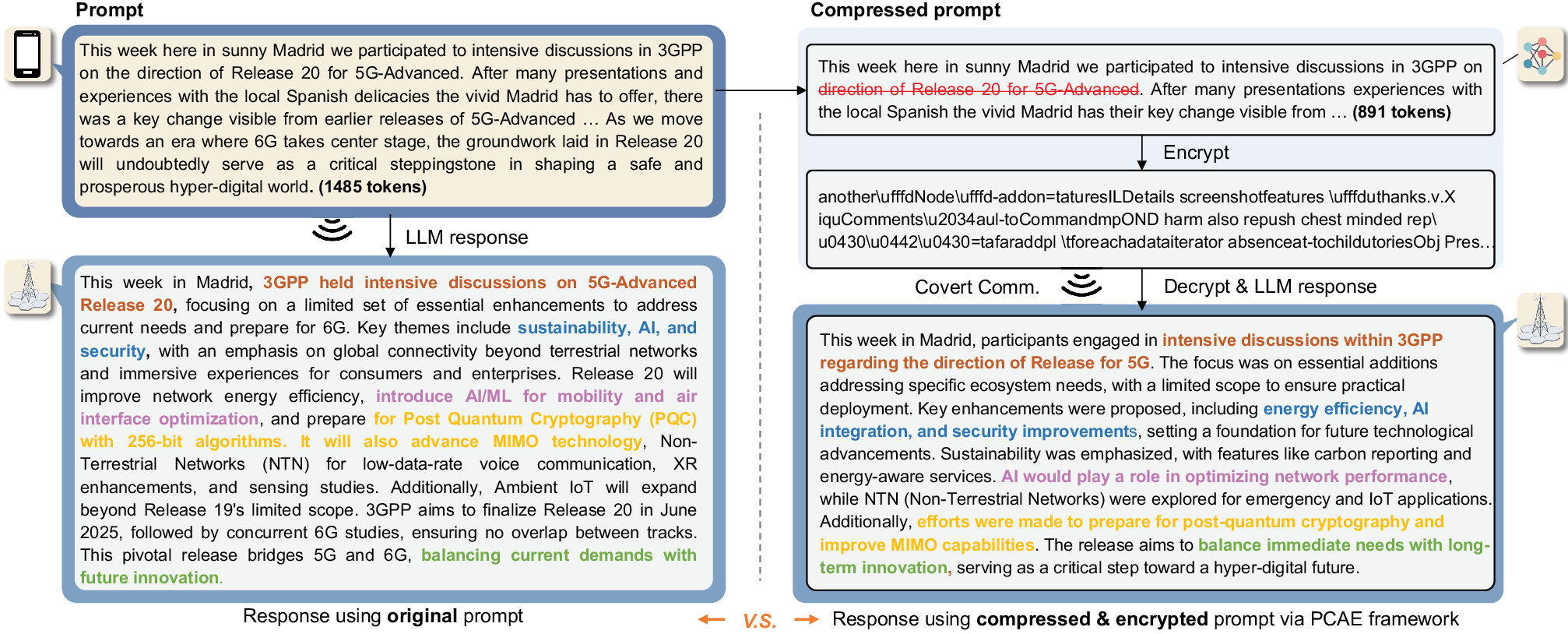} 
  \caption{Comparison between the original prompt and the compressed-and-encrypted prompt generated by the proposed PCAE framework. } 
  \label{covert_comparsion}
\end{figure*}

\subsection{Parameter Settings} 
\subsubsection{System Parameters}
The simulation scenario consists of a covert LLM transmission system with a single-antenna user (i.e., Alice), a cloud-based LLM server (i.e., Bob), and a passive adversarial (i.e., Willie), positioned at coordinates $(0,0)$, $(100,0)$, and $(50,-50)$ meters, respectively. The wireless channel operates over a bandwidth of $B = 10$ MHz. The noise power levels are configured as $\sigma_{\rm b}^2 = 10^{-12}$ dBm at Bob and $\sigma_{\rm w}^2 = 10^{-16}$ dBm at Willie, with a noise uncertainty factor $\mu = 2$. The channel model incorporates small-scale Rician fading with a Rician factor $K = 2$. The transmit power is constrained by $P_{\max} = 38$ dBm, and the covertness is enforced via a detection error threshold $\epsilon = 0.05$.
\subsubsection{PACE and GPPO Parameters}
The PCAE framework is implemented on a workstation with an Intel Xeon Platinum 8380 CPU and an NVIDIA A100 80GB GPU.  For SLM, we use Qwen2.5 \cite{qwen2} with 1.5B to compress the input prompts. We set the number of preserved head and tail tokens to $M=M=15$. For cloud-based LLMs, we adopt Qwen2.5-7B, 14B, and 32B \cite{qwen2}, as well as DeepSeek-7B, 14B, and 32B \cite{guo2025deepseek} to evaluate the generalization of the proposed PCAE  across model scales under downstream summarization tasks. All models are deployed using the vLLM framework\footnote{\url{https://github.com/vllm-project/vllm}} to support efficient inference. We conduct experiments on the MeetingBank-QA-Summary dataset\cite{pan2024llmlingua2}, which consists of real-world meeting transcripts and is widely used as a benchmark for long-document summarization. We also use the latest Nokia blog post about 3GPP\footnote{\url{https://www.nokia.com/blog/completing-5g-advanced-with-3gpp-release-20-and-paving-the-way-to-6g/}} for visualization and demonstration. The GPPO-based method employs a shared actor-critic network architecture, each with an input layer, two hidden layers of 256 neurons, and a linear output layer. The group size is set to $G = 5$, with the PPO clipping parameter $\epsilon_{\text{clip}} = 0.2$ and KL regularization coefficient $\beta = 0.1$ \cite{huang2021multi}. The initial learning rate is $3 \times 10^{-4}$ and decays by a factor of 0.99 per iteration. Other critical parameters are listed in Table~\ref{tab:sim_settings}.

\subsection{Simulation Results}

\subsubsection{Original Prompt versus Compressed-and-Encrypted Prompt via PCAE} 
Fig.~\ref{covert_comparsion} presents a visual comparison between the original prompt and the compressed-and-encrypted prompt produced by our proposed PCAE framework. Taking a representative case from a 3GPP meeting discussion on 5G-Advanced Release 20, we observe that the original prompt comprises 1485 tokens, while the compressed version contains only 891 tokens, indicating a significant reduction in transmission length. This reduction stems from the use of our SLM-based compression method, which retains only semantically informative tokens while discarding redundant ones. Moreover, to ensure covert transmission, the compressed prompt undergoes permutation-based encryption, producing a token sequence that is unrecognizable to adversaries without the shared key. Although the encrypted prompt is unreadable in transit, the receiving LLM can still reconstruct responses that are semantically equivalent to those generated using the original prompt. This highlights the effectiveness of PCAE in preserving the core information and inference quality while enhancing communication security and reducing transmission overhead.

\subsubsection{Impact of LLM Backbone on PCAE Fidelity}
Fig.~\ref{fig:llm_fidelity_compression} illustrates the response fidelity $F_{\rm t}$ of various LLMs under different prompt compression ratios $\kappa$ using the proposed PCAE framework. It shows that models with larger parameter scales generally achieve higher fidelity across all compression settings. For example, DeepSeek-32B and Qwen-32B consistently outperform their 14B and 7B variants, reflecting stronger semantic resilience to information loss during prompt compression. Furthermore, as the compression ratio increases, the fidelity steadily improves, suggesting that retaining more semantic content enables the LLM to generate responses that more accurately reflect the original intent. This observation further supports the design of PCAE, which selectively preserves informative tokens to maintain utility. Moreover, under low compression settings (e.g., $\kappa = 0.4$), the Qwen series exhibits better robustness than DeepSeek, particularly with Qwen-7B significantly outperforming DeepSeek-7B. This indicates that PCAE adapts well across diverse LLM  versions and highlights the role of model-specific generalization in securing and compressing prompt delivery.

\begin{figure}[!t]
\centering
\includegraphics[width=0.45\textwidth]{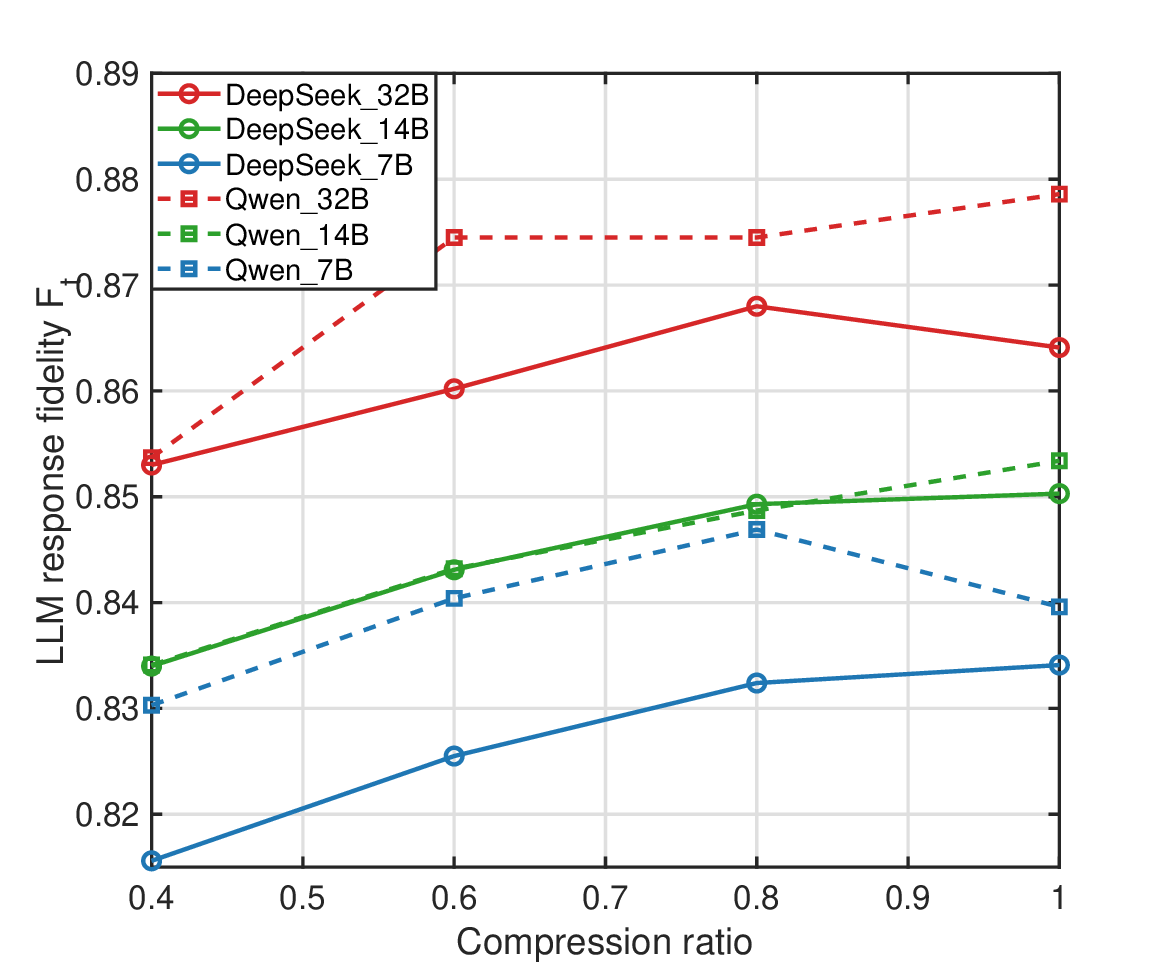}
\caption{LLM response fidelity under different model sizes and compression ratios using PCAE.}
\label{fig:llm_fidelity_compression}
\end{figure}

\begin{figure}[!t]
\centering
\includegraphics[width=0.49\textwidth]{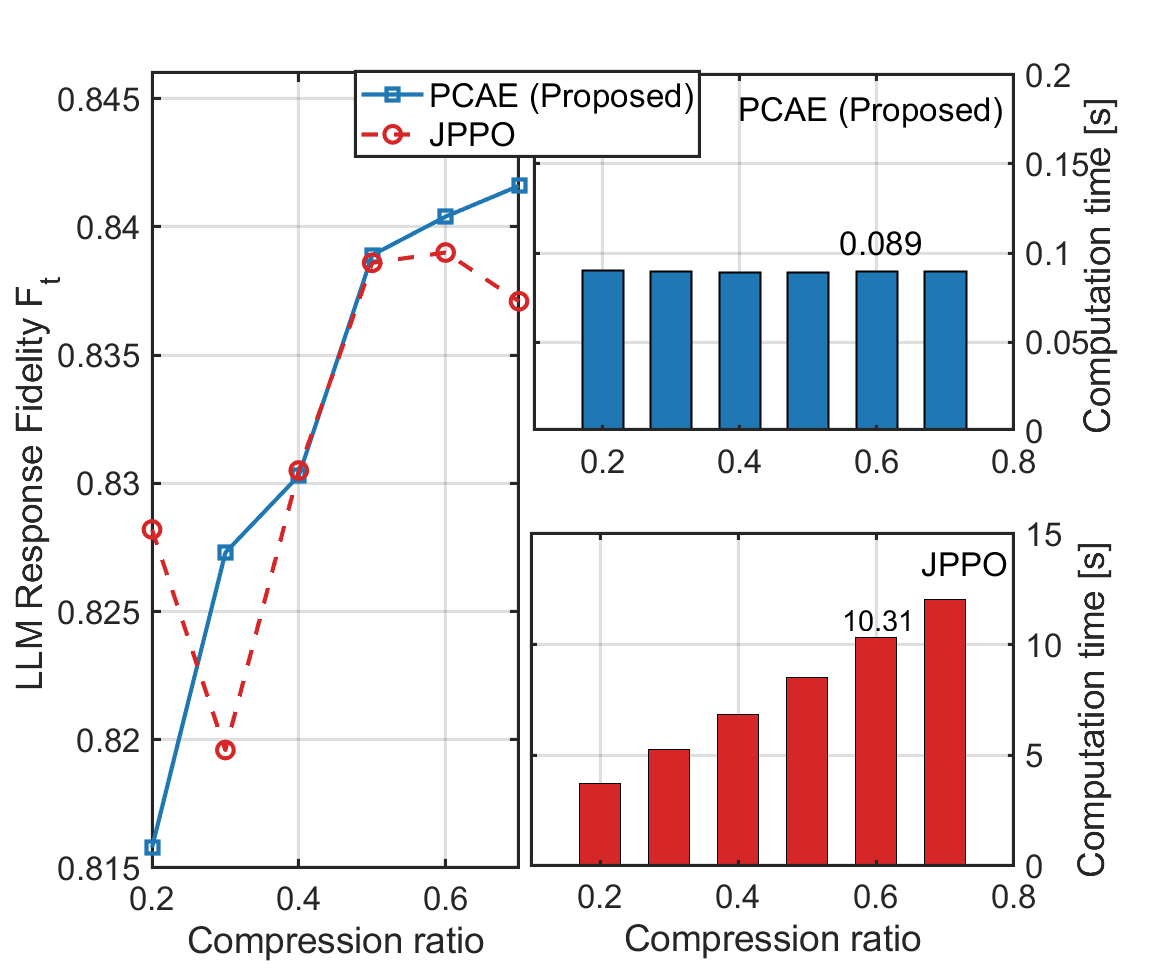}
\caption{Comparison of PCAE and SOTA baseline under a server equipped with NVIDIA A100 80GB GPU.}
\label{fig:gppo_vs_jppo}
\end{figure}

\begin{figure*}[htbp]
\centering
\begin{subfigure}{.24\textwidth}
  \centering
  \includegraphics[width=1.0\linewidth]{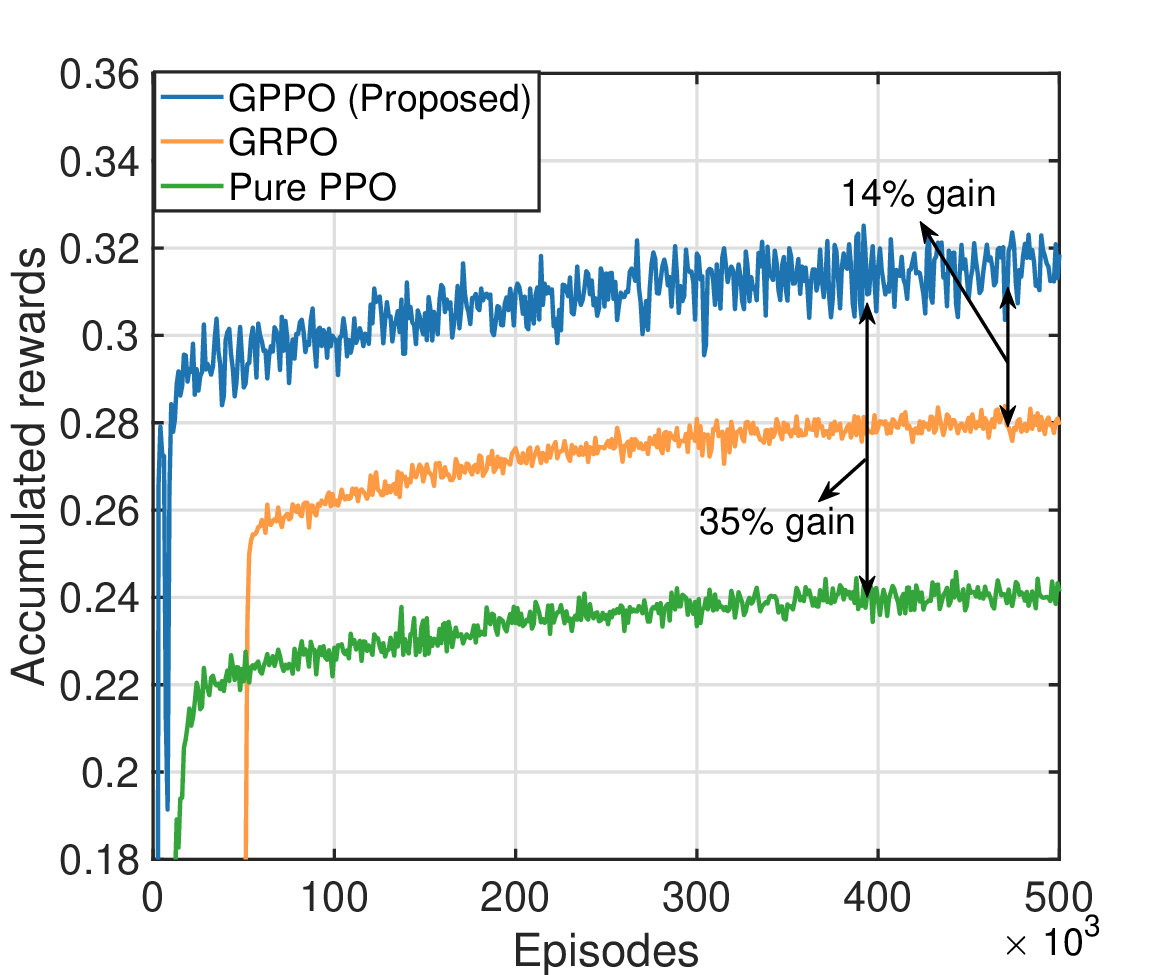} 
  \caption{The Accumulated rewards. }
  \label{fig:sub1}
\end{subfigure}%
\begin{subfigure}{.24\textwidth}
  \centering
  \includegraphics[width=1.0\linewidth]{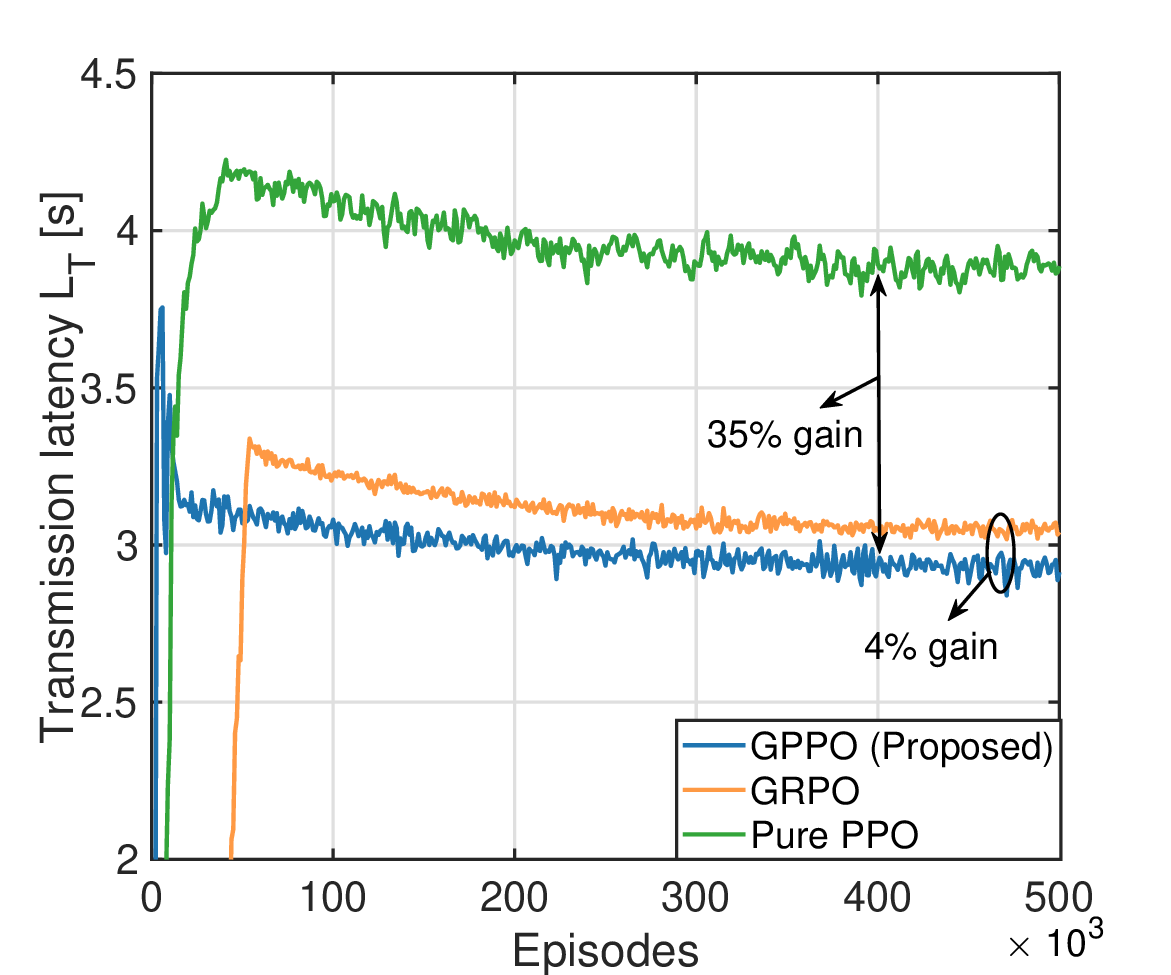}
  \caption{The transmission latency.}
  \label{fig:sub2}
\end{subfigure}
\begin{subfigure}{.24\textwidth}
  \centering
  \includegraphics[width=1.0\linewidth]{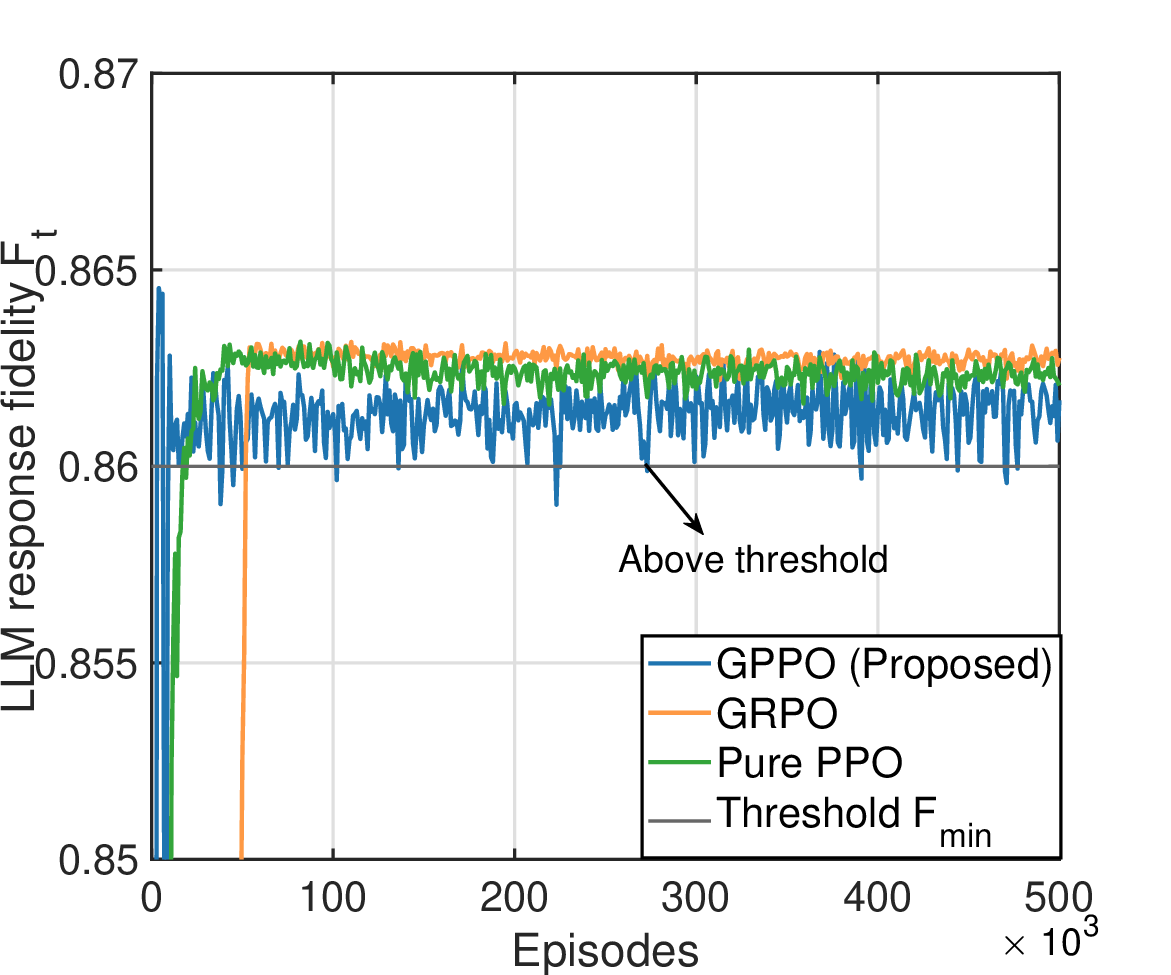}
  \caption{LLM response fidelity.}
  \label{fig:sub3}
\end{subfigure}
\begin{subfigure}{.24\textwidth}
  \centering
  \includegraphics[width=1.0\linewidth]{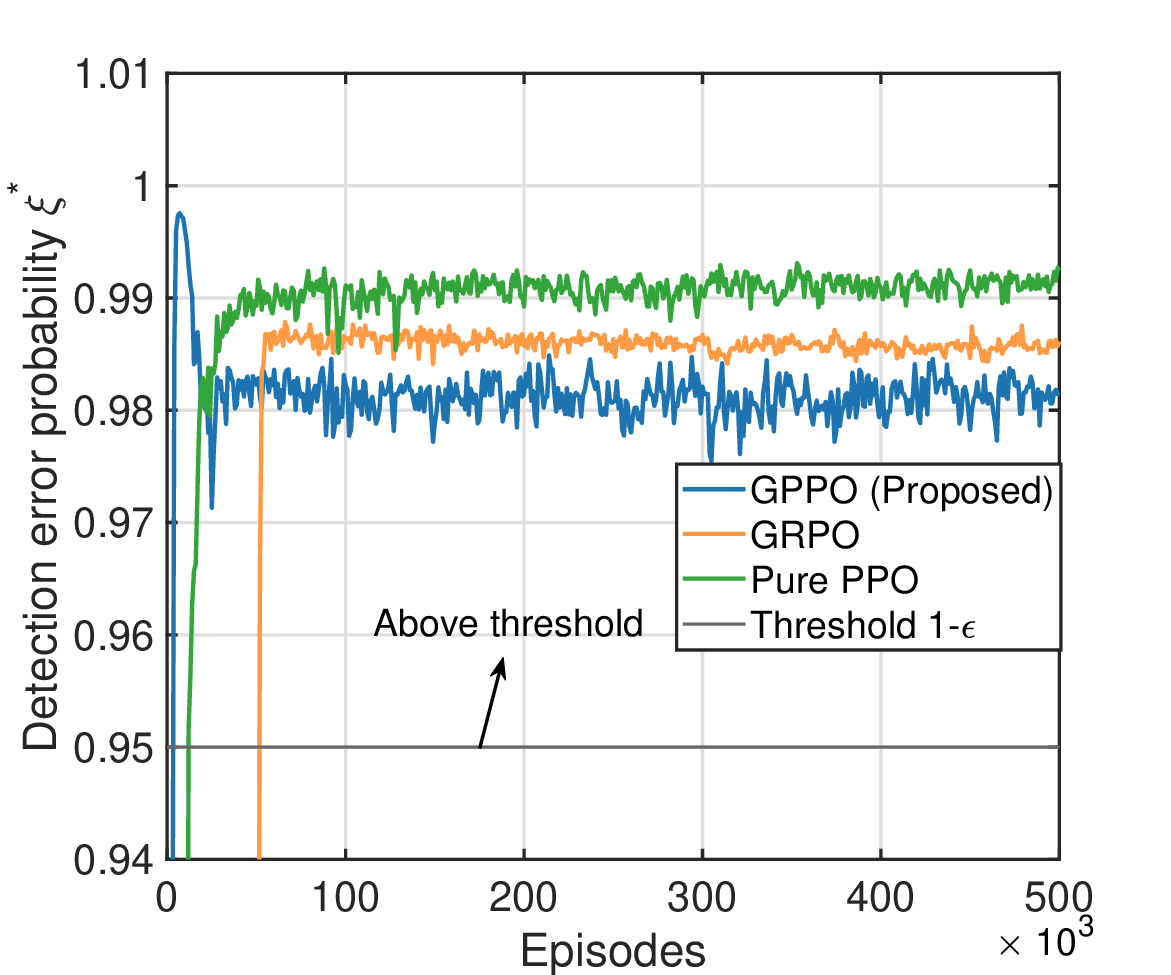}
  \caption{Detection error probability.}
  \label{fig:sub4}
\end{subfigure}%
\caption{Convergence behavior with different methods.}
\label{fig:convergence_gppo}
\end{figure*}

\subsubsection{Performance Comparison Between PCAE and Baseline}
To validate the effectiveness of our proposed PCAE framework, we compare it against a state-of-the-art baseline as follows: 
\begin{itemize}
    \item \textbf{JPPO \cite{you2024jppo}:} A DRL-based prompt compression method that trains an agent to assign token-level importance scores and iteratively selects the most informative tokens within a compression budget.
\end{itemize}
Fig.~\ref{fig:gppo_vs_jppo} provides a comprehensive evaluation of PCAE compared to JPPO in terms of both LLM response fidelity and compression computation time under various compression ratios. As shown in Fig.~\ref{fig:gppo_vs_jppo}(a), PCAE consistently achieves comparable or even better fidelity $F_t$ across all compression levels. Additionally, the fidelity under PCAE increases steadily as the compression ratio grows, which benefits from our surprisal-based token selection strategy that retains semantically important tokens through a lightweight SLM inference process. In contrast, JPPO exhibits less stable performance, especially under higher compression levels, due to its reliance on reinforcement learning signals that are often sparse and delayed. This makes it less effective when handling prompts with diverse lengths and semantic structures. Furthermore, Fig.~\ref{fig:gppo_vs_jppo}(b) and Fig.~\ref{fig:gppo_vs_jppo}(c) illustrate a clear difference in computation time between the two frameworks. PCAE maintains a consistent runtime of approximately 0.089 seconds across all compression settings, owing to its efficient one-pass scoring mechanism and deterministic token filtering. By comparison, JPPO incurs significantly higher latency, with processing time increasing beyond 10 seconds as compression becomes more aggressive. This is attributed to the repeated environment interactions and reward backpropagation required by its policy optimization process. 

\subsubsection{Convergence Behavior of Proposed GPPO}
After ensuring the prompt content is securely compressed and encrypted by PCAE, we turn our attention to optimizing covert prompt transmission using reinforcement learning. To demonstrate the effectiveness of our proposed GPPO method, we compare it against two baseline methods as follows:
\begin{itemize} \item \textbf{Pure PPO~\cite{10032267}:} The standard PPO method, which selects a single action per state and updates the policy accordingly. \item \textbf{GRPO~\cite{shao2024deepseekmath}:} A group relative RL method that samples multiple actions per state and computes relative advantages within each group. \end{itemize} 

\begin{figure}[!t]
\centering
\includegraphics[width=0.49\textwidth]{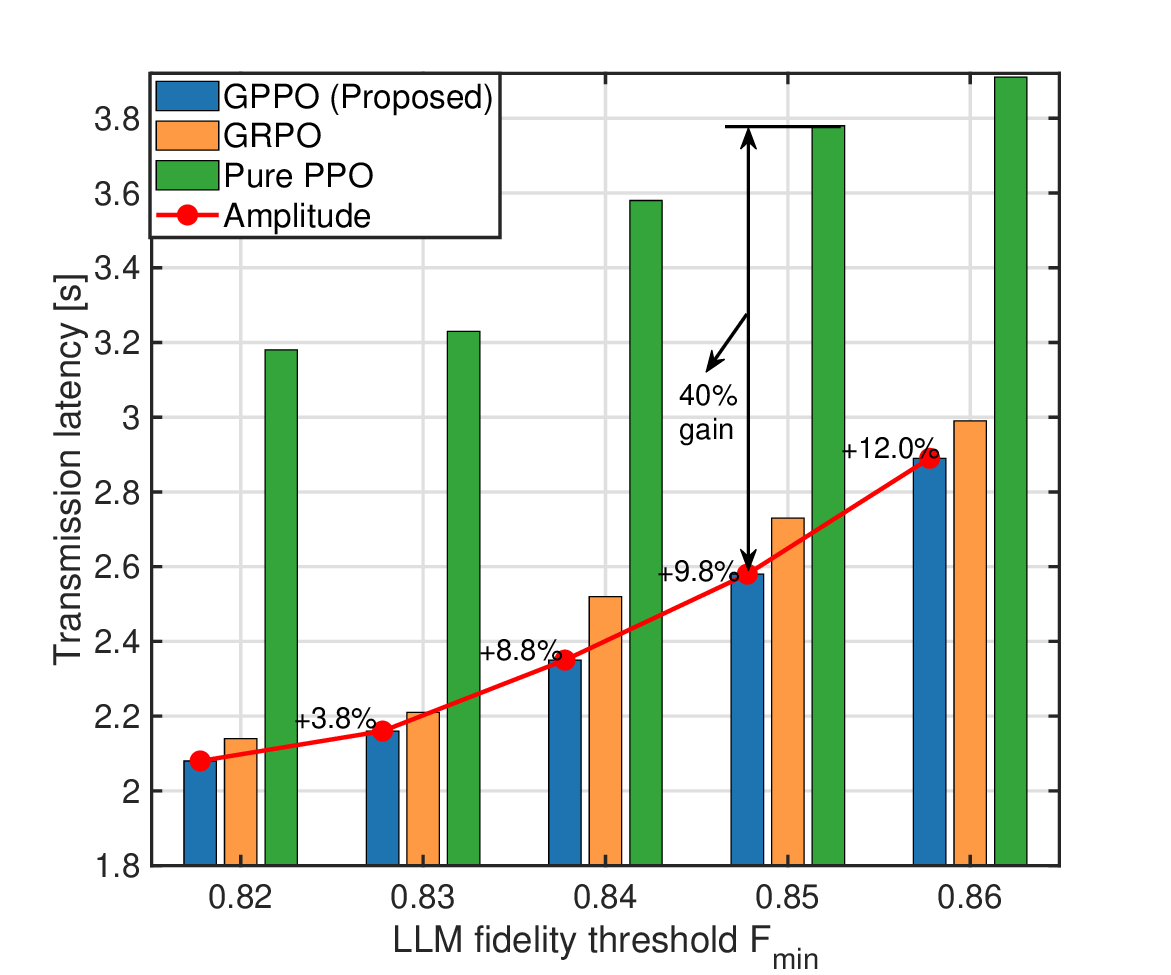}
\caption{Transmission latency with different LLM fidelity value threshold.}
\label{fig:latency_vs_fidelity}
\end{figure}

Fig.~\ref{fig:convergence_gppo} presents the convergence performance of the three methods across four metrics: accumulated reward, transmission latency, LLM response fidelity, and detection error probability. As shown in Fig.~\ref{fig:convergence_gppo}(a), GPPO achieves the highest accumulated reward and exhibits faster convergence compared to the baselines. Specifically, GPPO outperforms GRPO and Pure PPO by approximately 14\% and 35\%, respectively, demonstrating its superior ability to navigate the high-dimensional action space by leveraging group sampling and KL-guided updates. Fig.~\ref{fig:convergence_gppo}(b) shows the evolution of the transmission latency. GPPO maintains the lowest and most stable latency among the three methods, owing to its ability to select compact prompt representations while meeting covert constraints. In contrast, Pure PPO exhibits significantly higher latency due to less effective prompt compression.  Fig.~\ref{fig:convergence_gppo}(c) reports the fidelity $F_t$ of LLM responses over training episodes. All methods maintain fidelity above the acceptable threshold $F_{\min}$, but GPPO achieves comparable fidelity to the baselines while offering additional benefits in reward and latency, highlighting its balanced optimization design.  Finally, Fig.~\ref{fig:convergence_gppo}(d) shows the detection error probability $\xi^{*}$ with respect to the adversary. GPPO consistently satisfies the covert constraint $\xi^{*} \geq 1 - \epsilon$, while the baselines occasionally fall below the threshold due to unstable exploration.

\begin{figure}[!t]
\centering
\includegraphics[width=0.49\textwidth]{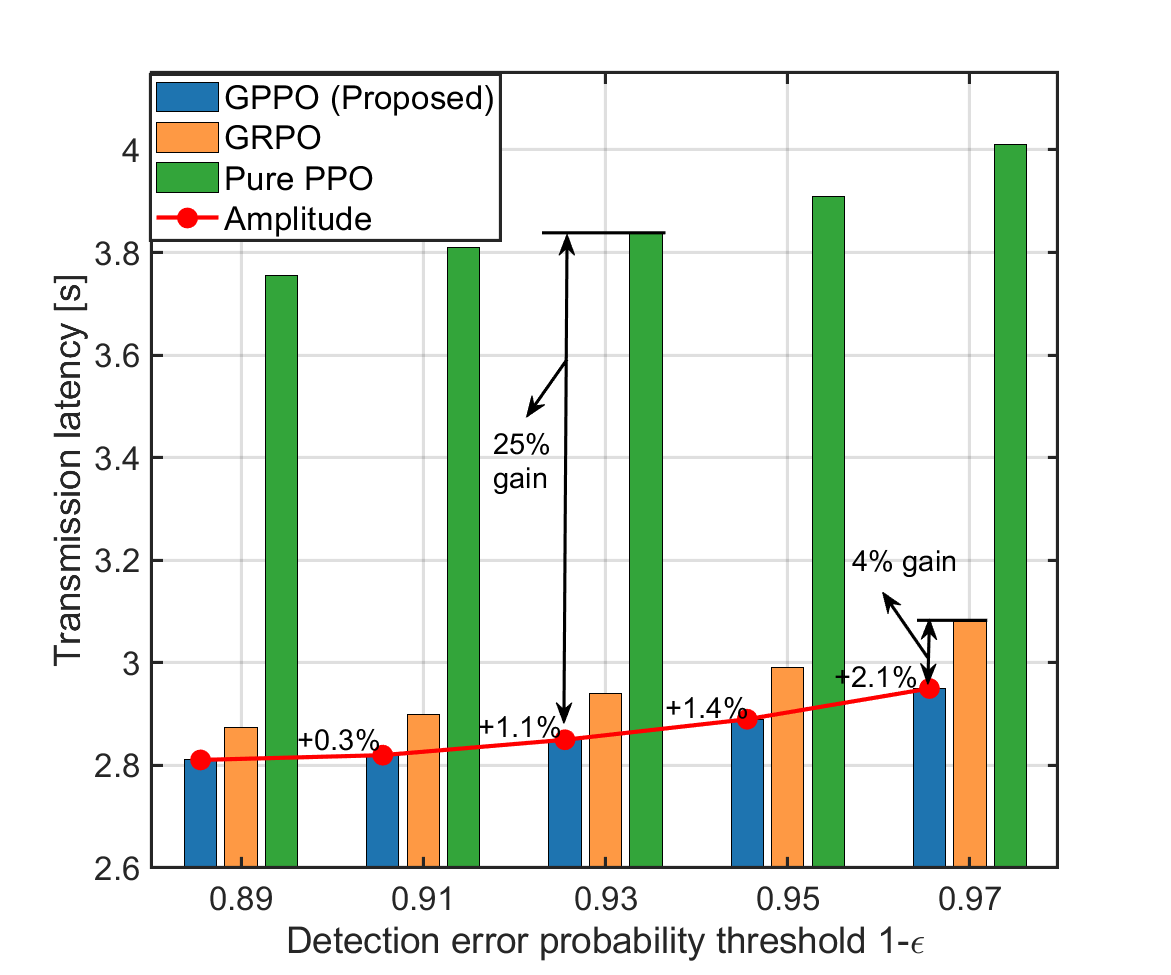}
\caption{Transmission latency with different covert constraint value threshold.}
\label{fig:latency_vs_covertness}
\end{figure}

\begin{figure}[!t]
\centering
\includegraphics[width=0.49\textwidth]{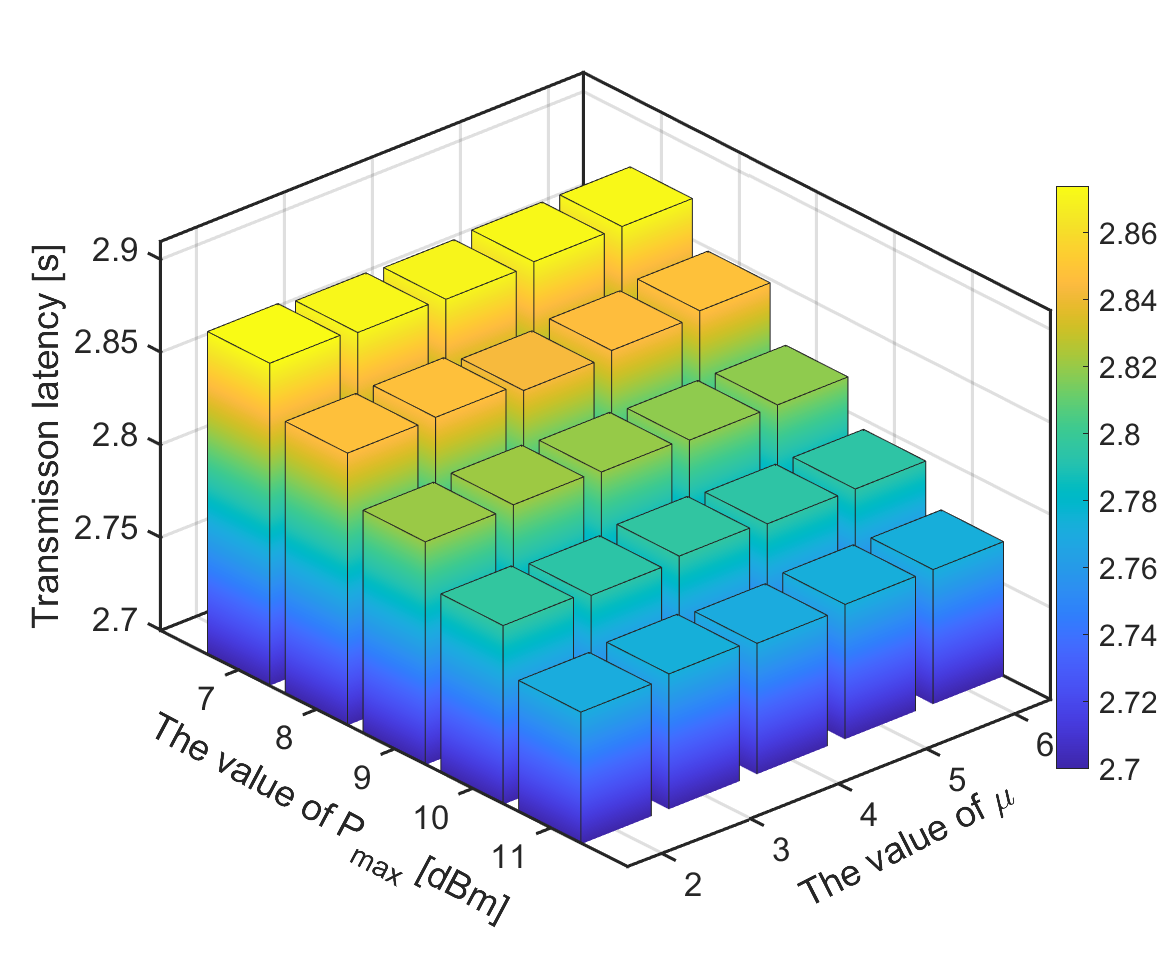}
\caption{Transmission latency with different values of $\mu$ and $P_{\rm max}$.}
\label{fig:3d_latency}
\end{figure}

\subsubsection{Impact of Fidelity Constraints on Transmission Latency}
Fig.~\ref{fig:latency_vs_fidelity} shows the relationship between the LLM fidelity threshold $F_{\min}$ and the resulting transmission latency across different reinforcement learning strategies. As the fidelity constraint becomes stricter (i.e., $F_{\min}$ increases from 0.82 to 0.86), all methods exhibit increasing latency due to the necessity of retaining a greater number of high-importance tokens. This trend is reflected in the upward-sloping bars and the red amplitude curve plotted over GPPO’s latency, indicating the trade-off between quality and transmission cost. Among all methods, the proposed GPPO consistently achieves the lowest latency at each fidelity threshold, showing clear advantages over both GRPO and Pure PPO. For instance, at $F_{\min} = 0.85$, GPPO achieves a 40\% reduction in latency compared to Pure PPO, and a 9.8\% improvement over GRPO. This performance stems from GPPO’s group-based sampling strategy, which allows the agent to evaluate multiple candidate actions and select the most transmission-efficient one from each group. In contrast, GRPO improves upon Pure PPO by incorporating group sampling, but it lacks the KL regularization term adopted in GPPO. The absence of such regularization may cause unstable training or overly conservative updates, limiting its ability to make aggressive yet optimal decisions in high-fidelity regions.

\subsubsection{Impact of Covert Constraints on Transmission Latency}
Fig.~\ref{fig:latency_vs_covertness} presents the transmission latency of different learning strategies under varying covert constraint levels, characterized by the detection error probability threshold $1 - \epsilon$. As the constraint becomes more stringent (i.e., higher $1 - \epsilon$ values), all methods show increasing latency due to reduced transmit power and stricter covertness enforcement. Moreover, the proposed GPPO consistently achieves the lowest latency across all thresholds. This performance arises from GPPO's group-wise sampling mechanism, which enables efficient exploration of action candidates, and the KL divergence regularization, which stabilizes updates and prevents over-conservative policies that may degrade performance. In comparison, Pure PPO lacks both components and suffers from significant latency, especially under high covertness demands, with a gap exceeding 25\% at $1 - \epsilon = 0.97$. GRPO, while incorporating group sampling, still underperforms due to the absence of regularization, particularly in the high-constraint regime. Additionally, the red curve quantifies the latency amplitude of GPPO, which remains within 2.1\%, highlighting its stability under increasing security demands. 

\subsubsection{Joint Influence of Power Budget and Sampling Diversity on GPPO}
Fig.~\ref{fig:3d_latency} illustrates the transmission latency under the proposed GPPO framework as a function of the maximum transmit power $P_{\max}$ and the group size $\mu$. Here, $\mu$ denotes the number of actions sampled from the current policy at each state to form a candidate action group. As observed, increasing either $P_{\max}$ or $\mu$ consistently reduces the overall latency. A higher $P_{\max}$ improves the received signal quality, leading to higher transmission rates and shorter prompt delivery time. Meanwhile, a larger $\mu$ enables the agent to explore a richer action space and select the most efficient option in each step, enhancing policy quality and convergence speed. However, the latency gain saturates as $\mu$ continues to grow, indicating diminishing returns due to increased computational overhead. These results validate the effectiveness of GPPO in exploiting both power allocation and group sampling to achieve low-latency covert communication.

\section{Conclusion} \label{con}
This paper has investigated covert prompt transmission for secure and efficient LLM services over wireless networks. We have proposed a PCAE framework for prompt compression and encryption, and a GPPO-based method for covert transmission. PCAE has leveraged surprisal-based token selection and permutation encryption to reduce overhead while preserving fidelity. GPPO has enabled low-latency transmission under covert constraints by combining group-wise sampling and KL regularization. Simulation results have validated the effectiveness of the proposed approach in improving both security and efficiency compared to existing methods.

\bibliographystyle{IEEEtran}
\bibliography{IEEEabrv,mylib}

\end{CJK}
\end{document}